\documentclass[sn-mathphys-num]{sn-jnl}

\usepackage[T1]{fontenc}
\usepackage{amsmath,amssymb,amsthm,mathtools}
\usepackage{mathrsfs}
\usepackage{enumitem}
\usepackage{xcolor}
\usepackage{hyperref}
\usepackage[nameinlink,capitalise]{cleveref}
\usepackage{booktabs}
\usepackage{array}
\usepackage{colortbl}
\usepackage{graphicx}
\usepackage{manyfoot}
\usepackage{tikz}
\usetikzlibrary{arrows.meta,positioning,calc,backgrounds}

\hypersetup{
  colorlinks=true,
  linkcolor=blue,
  citecolor=blue,
  urlcolor=blue
}

\allowdisplaybreaks
\theoremstyle{plain}
\newtheorem{theorem}{Theorem}[section]
\newtheorem{proposition}[theorem]{Proposition}
\newtheorem{lemma}[theorem]{Lemma}
\newtheorem{corollary}[theorem]{Corollary}

\theoremstyle{definition}
\newtheorem{definition}[theorem]{Definition}
\newtheorem{remark}[theorem]{Remark}

\newcommand{\mb}{\mathbb}
\newcommand{\mc}{\mathcal}

\newcommand{\Tr}{\operatorname{Tr}}
\newcommand{\tr}{\operatorname{tr}}
\newcommand{\id}{\operatorname{id}}

\newcommand{\one}{\mathbf 1}
\newcommand{\rank}{\operatorname{rank}}

\newcommand{\spec}{\operatorname{spec}}

\newcommand{\Dens}{\mc D}
\newcommand{\St}{\operatorname{St}}

\newcommand{\ket}[1]{\lvert #1\rangle}
\newcommand{\bra}[1]{\langle #1\rvert}
\newcommand{\ketbra}[2]{\lvert #1\rangle\!\langle #2\rvert}
\newcommand{\ip}[2]{\left\langle #1,#2\right\rangle}
\newcommand{\cH}{\mc H}

\newcommand{\MOE}{H_{\min}}

\title{Superadditivity of classical communication over quantum channels via random and deterministic permutations}

\author[1]{\fnm{Benjamin} \sur{Lovitz}}
\author[2,3]{\fnm{Peixue} \sur{Wu}}
\affil[1]{\orgdiv{Department of Computer Science and Software Engineering},
  \orgname{Concordia University},
  \orgaddress{\city{Montr\'eal}, \state{Quebec}, \country{Canada}}}
\affil[2]{\orgdiv{Department of Mathematics, Institute for Quantum and Information Sciences}, \orgname{Syracuse University},
  \orgaddress{\city{Syracuse}, \state{New York}, \country{USA}}}
\affil[3]{\orgdiv{Institute for Quantum Computing}, \orgname{University of Waterloo},
  \orgaddress{\city{Waterloo}, \state{Ontario}, \country{Canada}}}

\abstract{
Since Hastings' proof of superadditivity of classical communication over quantum channels, 
considerable effort has been devoted to finding a structural explanation of this phenomenon 
that was originally established by concentration of measure for Haar random unitaries. 

The main observation of this work is that Haar randomness can be replaced by random permutations
without changing the limiting geometry responsible for nonadditivity.
This replacement turns a continuous problem over unitary matrices
into a discrete combinatorial problem over zero--one permutation matrices,
and thereby opens a path toward derandomization.

The theorem of Bordenave and Collins shows that random permutations have the
required limiting behavior and the algorithm of O'Donnell and Wu then provides a
deterministic asymptotic construction, running in polynomial time in the
size when the channel parameters and accuracy are fixed. Thus the random
construction can be derandomized in an asymptotic algorithmic sense.

Finally, a quantitative random permutation estimate by Chen, Garza-Vargas, Tropp and van Handel
gives a fully numerical estimate: there exists a tuple of 57,836,025 permutations acting on a set of size 
\[ N \le 5.422\times 10^{116216}\]
such that the associated finite dimensional channel exhibits nonadditivity.
This enormous value remains an obstacle to a practical construction.
}

\begin{document}
\maketitle

\clearpage
\tableofcontents

\section{Introduction}\label{sec:introduction}
 
A fundamental difference between classical and quantum noise is that two independent quantum
channels can jointly produce less output entropy than the sum of what each produces on
its own. Let \(\Phi:\Dens(\cH_A)\to\Dens(\cH_B)\) be a finite-dimensional quantum
channel. Its minimum output von Neumann entropy is
\begin{equation}\label{eq:intro-moe}
    \MOE(\Phi)
    :=
    \inf_{\rho\in\Dens(\cH_A)} H(\Phi(\rho)),
    \qquad
    H(\sigma):=-\Tr(\sigma\log\sigma).
\end{equation}
For two channels, product inputs always give
\begin{equation}\label{eq:intro-subadditivity}
    \MOE(\Phi_1\otimes\Phi_2)
    \leq
    \MOE(\Phi_1)+\MOE(\Phi_2).
\end{equation}
For classical channels, concavity of the entropy allows a minimizing input to be chosen
deterministic, and a deterministic input to a product channel produces a product output;
equality in \eqref{eq:intro-subadditivity} is therefore automatic. Quantum mechanics
introduces a new possibility, namely an entangled pure input. A strict inequality in
\eqref{eq:intro-subadditivity} thus detects genuinely collective behavior: two
independently acting channels jointly produce less entropy than any product input can
achieve.
 
This question is fundamental to classical communication over noisy quantum systems. For
a finite ensemble \(\{p_i,\rho_i\}\), the Holevo information at the output of \(\Phi\) is
\begin{equation}\label{eq:intro-holevo-ensemble}
    \chi(\{p_i,\rho_i\},\Phi)
    :=
    H\!\left(\Phi\!\left(\sum_i p_i\rho_i\right)\right)
    -\sum_i p_i H(\Phi(\rho_i)),
\end{equation}
and \(\chi_{\max}(\Phi)\) denotes its supremum over all finite ensembles. The
Holevo--Schumacher--Westmoreland theorem~\cite{Holevo1998,SchumacherWestmoreland1997}
identifies the unassisted classical capacity with the regularization
\begin{equation}\label{eq:intro-classical-capacity}
    C(\Phi)
    =
    \lim_{m\to\infty}\frac1m\chi_{\max}(\Phi^{\otimes m})
    =
    \sup_{m\geq1}\frac1m\chi_{\max}(\Phi^{\otimes m}).
\end{equation}
If the Holevo information were additive, the asymptotic optimization in
\eqref{eq:intro-classical-capacity} would collapse to a single use of the channel. Shor
proved that, as universal statements, additivity of the minimum output entropy,
additivity of the Holevo information, additivity of the entanglement of formation, and
strong superadditivity of the entanglement of formation are all
equivalent~\cite{Shor2004}. The additivity problem for the minimum output entropy is
therefore central to quantum information theory.
 
The first counterexamples concerned the minimum output \(p\)-R\'enyi entropies. Werner
and Holevo exhibited a channel violating the corresponding multiplicativity conjecture
for sufficiently large \(p\)~\cite{WernerHolevo2002}; Hayden and Winter obtained random
counterexamples for every \(p>1\)~\cite{HaydenWinter2008}; and Cubitt, Harrow, Leung,
Montanaro, and Winter treated \(p\) close to zero~\cite{CubittEtAl2008}. The case \(p=1\)
proved substantially more delicate. Hastings showed that a random finite-dimensional
channel \(\Phi_n\), paired with its entrywise complex conjugate \(\overline\Phi_n\),
typically satisfies~\cite{Hastings2009}
\begin{equation}\label{eq:intro-hastings}
    \MOE(\Phi_n\otimes\overline\Phi_n)
    <
    \MOE(\Phi_n)+\MOE(\overline\Phi_n)
    =2\MOE(\Phi_n).
\end{equation}
Subsequent work clarified the concentration-of-measure, geometric, and
free-probabilistic content of this
phenomenon~\cite{FukudaKingMoser2010,BrandaoHorodecki2010,FukudaKing2010,AubrunSzarekWerner2011,CollinsNechita2010,CollinsNechita2011,Fukuda2014,BelinschiCollinsNechita2016,aubrun2017alice}.
A common structure emerged: one-copy outputs are close to maximally mixed, whereas a
canonical entangled input to a conjugate pair produces an output with an anomalously
large eigenvalue.
 
Although existence is no longer the central question, the mechanism behind this
phenomenon and its structured finite-dimensional realization remain elusive. Concentration-of-measure arguments show that a random channel works with high
probability, but by themselves they neither explain why it works nor single out a
specific channel.
 
A modern framework for this problem is \emph{strong convergence}; see~\cite{vanhandel2026}
for an overview. Belinschi, Collins, and Nechita used strong convergence of Haar-random
projections to analyze channels induced by Haar-random
subspaces~\cite{BelinschiCollinsNechita2016}. Collins combined strong asymptotic
freeness with Haagerup's inequality to give a conceptual proof for random mixed-unitary
channels~\cite{Collins2018}, and Fukuda, Hasebe, and Sato used strong convergence to
semicircular and circular systems to construct further random
examples~\cite{FukudaHasebeSato2022}. These works suggest that new strong-convergence
models can lead to more structured counterexamples to additivity of the minimum output
entropy.
 
In this work we develop the strong-convergence approach to this problem systematically.
We place the free limits of random-channel constructions in a common framework, and we
combine the permutation strong-convergence theorem of Bordenave and
Collins~\cite{BordenaveCollins2019} with the deterministic lift construction of
O'Donnell and Wu~\cite{ODonnellWu2020} to obtain both random and deterministic
permutation channels that violate additivity.
 
The argument passes through infinite-dimensional targets. We begin with two standard
random-channel models: complementary channels of Haar-random mixed-unitary channels, as
in~\cite{Hastings2009}, and channels induced by Haar-random subspaces, as
in~\cite{AubrunSzarekWerner2010,AubrunSzarekWerner2011}. Strong asymptotic freeness of
Haar unitaries~\cite{CollinsMale2014} identifies their large-environment limits as,
respectively, a free mixed-unitary complementary channel and a free-compression channel;
see also~\cite{CollinsFukudaNechita2015}. Both limiting channels are exactly solvable:
their one-copy output bodies, that is, their sets of output states, are determined by free
spectral calculations, while a canonical infinite-dimensional analogue of the maximally
entangled state produces an explicit low-entropy output for the product with the
conjugate channel. In suitable parameter regimes, the resulting entropy gap is strictly
positive.
 
For the finite-dimensional realization, we concentrate on the random-subspace model and
its free-compression limit; the mixed-unitary model admits a parallel treatment, but the
subspace model leads to more favorable quantitative parameters. Our construction
separates the two ingredients needed for nonadditivity. The one-copy output body is
controlled by the spectral behavior of a tuple of permutations, whereas an auxiliary
finite phase system forces the desired two-copy output exactly, for every choice of the
permutations. The one-copy spectral approximation and the two-copy entangled-input
calculation are thus handled by independent mechanisms.
 
It then remains to choose the permutations so that the one-copy output body
approximates that of the free limit. By Bordenave--Collins strong
convergence~\cite{BordenaveCollins2019}, independent uniform permutations do so with
probability tending to one, and the construction of O'Donnell and
Wu~\cite{ODonnellWu2020} produces suitable permutations deterministically. Our main
result is the following.
 
\begin{theorem}[Main result, informal]\label{thm:intro-main-informal}
There exist a deterministic channel \(\Phi_\infty\) with infinite-dimensional input and
finite-dimensional output, and a bipartite output state \(\rho_\infty\) of
\(\Phi_\infty\otimes_{\max}\overline{\Phi}_\infty\), such that
\begin{equation}\label{eq:nonadditivity-infinite}
    H(\rho_\infty)<2\MOE(\Phi_\infty).
\end{equation}
For each \(N\ge 2\), there is a finite-dimensional channel \(\Phi_N\), generated by a
tuple of permutations of \(N\) elements, such that
\(\MOE(\Phi_N)\to\MOE(\Phi_\infty)\) and, for all sufficiently large \(N\),
\begin{equation}\label{eq:nonadditivity-large-N}
    \MOE(\Phi_N\otimes\overline \Phi_N)
    <2\MOE(\Phi_N).
\end{equation}
Moreover, for all sufficiently large \(N\), a channel \(\Phi_N\) satisfying
\eqref{eq:nonadditivity-large-N} is produced by a deterministic algorithm running in
time polynomial in \(N\). Quantitatively, one can take \(57{,}836{,}025\) permutations
acting on a set of size
\[
N\le 5.422\times 10^{116216}.
\]
\end{theorem}
 
The quantitative statement relies on the effective strong-convergence bounds of Chen,
Garza-Vargas, Tropp, and van Handel~\cite{ChenGarzaVargasTroppVanHandel2026}.
 
\subsection{Construction of the channel}\label{subsec:intro-construction}
 
Our construction is motivated by the strong-convergence approach. We first recall how
the free-probabilistic approach
of~\cite{BelinschiCollinsNechita2012,BelinschiCollinsNechita2016} establishes
nonadditivity of the minimum output entropy. Any quantum channel can be described by a
bipartite subspace \(\mc W\subseteq\mb C^k\otimes\mb C^d\): the first factor \(\mb C^k\)
is the output (Bob), the second factor \(\mb C^d\) is the environment, and the channel
maps a state supported on \(\mc W\) to its reduced state on Bob. The key idea is to
choose a Haar-random subspace \(\mc W_n\subseteq\mb C^k\otimes\mb C^{d_n}\) with
\(\dim\mc W_n/(kd_n)\to t\in(0,1)\). Let \(\Phi_n\) be the associated channel, and let
\(\overline\Phi_n\) be the channel whose Kraus operators are the complex conjugates of
those of \(\Phi_n\). Write \(\Dens_k\) for the set of \(k\)-dimensional states,
\(F(\rho,\sigma)=\|\sqrt\rho\sqrt\sigma\|_1^2\) for the (squared) fidelity, and
\(\psi_k^+\) for the maximally entangled state on \(\mb C^k\otimes\mb C^k\). Then, almost
surely~\cite{BelinschiCollinsNechita2016},
\begin{align}
    & \lim_{n\to \infty} \MOE(\Phi_n) = \min_{\rho \in \Dens_k,\ F(\rho, \frac{I_k}{k}) \ge 1-t} H(\rho), \label{eq:1-one-copy}\\
    & \limsup_{n\to \infty} \MOE(\Phi_n \otimes \overline \Phi_n) \le H(\rho_{k,t}),
    \qquad \rho_{k,t}:=t\,\psi_k^+ +(1-t)\frac{I_{k^2}}{k^2}. \label{eq:1-two-copy}
\end{align}
For suitable \(k\) and \(t\), twice the right-hand side of~\eqref{eq:1-one-copy} exceeds
\(H(\rho_{k,t})\), so \(\MOE(\Phi_n\otimes\overline\Phi_n)<2\MOE(\Phi_n)\) for all large
\(n\).
 
Letting \(n\to\infty\), one obtains a channel \(\Phi_\infty\) with infinite-dimensional
input and finite-dimensional output satisfying~\eqref{eq:nonadditivity-infinite}; the
two-copy state \(\rho_\infty\) is described in the commuting-operator framework, or
equivalently on a maximal tensor product (see Section~\ref{subsec:von-neumann-basics}).
This raises the question that motivates our approach:
\begin{center}
\emph{Can the infinite-dimensional model be approximated by finite-dimensional channels
while preserving nonadditivity?}
\end{center}
A universal approximation of this kind cannot be expected. The underlying
infinite-dimensional algebra is a non-injective von Neumann algebra of type
\(\mathrm{II}_1\), and the celebrated result \(\mathrm{MIP}^{*}=\mathrm{RE}\) implies that
some commuting-operator correlations cannot be approximated by finite-dimensional
tensor-product correlations~\cite{JiEtAl2022}; as remarked in~\cite{CollinsYoun2022},
this rules out a universal finite-dimensional construction by approximation.
 
We therefore approximate only what nonadditivity actually requires. The one-copy
estimate~\eqref{eq:1-one-copy} must hold for every input; as explained in
Section~\ref{subsec:intro-methods}, it is governed by the largest eigenvalues of a fixed
family of matrix-valued polynomials of degree two in the unitaries of our construction.
The two-copy estimate~\eqref{eq:1-two-copy} concerns a single maximally entangled
input, and our construction reproduces it exactly. Strong convergence, which controls
the norms of \emph{all} polynomials, can therefore be relaxed to a uniform spectral
approximation of these degree-two polynomials. This weaker requirement is met by the
deterministic spectral-approximation algorithm of~\cite{ODonnellWu2020}, which outputs
suitable permutation tuples in time polynomial in \(N\) for fixed \(k\), \(m\), \(q\),
and approximation accuracy.
 
We now specify the construction. Define the Weyl--Bell basis of
\(\mb C^k\otimes\mb C^k\) by
\begin{equation}\label{eq:Weyl-basis-intro}
    |\alpha\rangle_k:=\frac{1}{\sqrt{k}}\sum_{j=0}^{k-1}\zeta_k^{bj}|j+a\rangle\otimes|j\rangle,
    \qquad \zeta_k=\exp(2\pi\mathrm{i}/k),
    \qquad \alpha=(a,b)\in\mb Z_k^2;
\end{equation}
it is obtained by vectorizing the generalized Pauli matrices. Fix integers \(k,m\ge2\)
and \(q\in\{1,\ldots,m-1\}\), and put \(\Theta=\mb Z_k^2\times\mb Z_m^2\). Let
\(\{|\xi\rangle:\xi\in\Theta\}\) be the Weyl--Bell basis of
\((\mb C^k\otimes\mb C^k)\otimes(\mb C^m\otimes\mb C^m)\), where
\(|\xi\rangle=|\alpha\rangle_k\otimes|\beta\rangle_m\) for \(\xi=(\alpha,\beta)\). For
\(\xi=(a,b,c,d)\in\Theta\), define the clock unitary
\begin{equation}\label{eq:intro-phase-system-short}
    R_\xi=Z_\ell^{(a+kb)(c+md)},\qquad
    Z_\ell := \sum_{j=0}^{\ell-1} \exp\!\Big(\frac{2\pi \mathrm{i}}{\ell}j\Big) |j\rangle \langle j|,\qquad
    \ell=(k^2-1)(m^2-1)+1.
\end{equation}
Given a tuple \(\boldsymbol U_N=(U_{\xi,N})_{\xi\in\Theta}\) of unitaries on a Hilbert
space \(\mc K_N\), define the controlled unitary
\begin{equation}\label{eq:intro-controlled-unitary-short}
    D_N=\sum_{\xi\in\Theta}|\xi\rangle\langle\xi|\otimes U_{\xi,N}\otimes R_\xi
\end{equation}
on \(\cH_N:=(\mb C^k\otimes\mb C^k)\otimes(\mb C^m\otimes\mb C^m)\otimes\mc K_N\otimes\mb C^\ell\).
Let \(p_q\in M_m\) be the projection onto the first \(q\) coordinate vectors, and put
\begin{equation}\label{eq:intro-projection-channel-short}
    Q=I_{k^2}\otimes(p_q\otimes I_m)\otimes I_{\mc K_N}\otimes I_\ell,
    \qquad
    P_N=D_NQD_N^*.
\end{equation}
Our channel is the one associated with the subspace \(\operatorname{ran}P_N\), where the
first tensor factor \(\mb C^k\) of \(\cH_N\) is the output and all remaining factors form
the environment. Equivalently, in the Heisenberg picture it is given by the unital
completely positive map
\begin{equation}\label{eq:intro-finite-ucp}
    \Gamma_N(A)=P_N(A\otimes I)P_N,\qquad A\in M_k,
\end{equation}
into the corner \(P_N\mc B(\cH_N)P_N\), where \(A\) acts on the first factor
\(\mb C^k\); we denote its dual (Schr\"odinger) channel by \(\Phi_N\). The
controlled-unitary structure of \(D_N\) is inspired by Ching's realization of free
products~\cite{Ching1973}.
 
For our permutation channels, we take \(\mc K_N=\one^\perp=\{x\in\mb C^N:\sum_jx_j=0\}\)
and \(U_{\xi,N}=S_{\xi,N}:=P_{\sigma_{\xi,N}}|_{\one^\perp}\), where \(P_\sigma\) denotes
the unitary induced by a permutation \(\sigma\) of \([N]\) and
\(\boldsymbol\sigma_N=(\sigma_{\xi,N})_{\xi\in\Theta}\) is a \(k^2m^2\)-tuple of
permutations. We write \(\Gamma_N^{\boldsymbol\sigma}\) and \(\Phi_N^{\boldsymbol\sigma}\)
for the resulting maps; every ingredient of \(\Phi_N^{\boldsymbol\sigma}\) is explicit
except the tuple \(\boldsymbol\sigma_N\).
 
The two families of unitaries play separate roles. The clock unitaries \(R_\xi\) force
the two-copy estimate~\eqref{eq:1-two-copy} exactly, for every choice of the
permutations, once the rank of the projection is fixed; see~\eqref{eq:intro-exact-bell-short}
below. The permutation unitaries \(S_{\xi,N}\) are chosen to reproduce the one-copy
estimate~\eqref{eq:1-one-copy} approximately: a deterministic algorithm running in time
polynomial in \(N\) outputs tuples \(\boldsymbol\sigma_N\) for which
\(\MOE(\Phi_N^{\boldsymbol\sigma})\to\MOE(\Phi_\infty)\) as \(N\to\infty\).
 
\subsection{Overview of the methods}\label{subsec:intro-methods}
 
The proof combines two estimates of different nature. The one-copy lower bound on the
minimum output entropy must hold for \emph{every} input; we will see
in~\eqref{eq:intro-moe-dual} that it is governed by the largest eigenvalues of the
Heisenberg observables of the channel, which is exactly what strong convergence
controls. The two-copy upper bound concerns a \emph{single} entangled input and is
governed by a Gram matrix of second moments, see~\eqref{eq:intro-gram-identity}; our
finite-dimensional construction reproduces it exactly. Accordingly, the proof proceeds in
two steps. First, by taking the limit of Haar-random channels, we identify an
infinite-dimensional free limit channel that violates additivity. Second, we construct
random and deterministic permutation channels whose one-copy output bodies approximate
that of the free limit, while their two-copy output state coincides with the limiting
one.
 
We begin with the form of strong convergence used in this paper, namely the
matrix-coefficient formulation emphasized in~\cite{vanhandel2026}. Throughout,
\(\|\cdot\|\) denotes the operator norm.
 
\begin{definition}[Strong convergence]\label{def:intro-strong-convergence}
Let \(\boldsymbol X_N=(X_{1,N},\ldots,X_{r,N})\) be random operator tuples in unital
\(C^*\)-algebras \(\mc A_N\), and let \(\boldsymbol x=(x_1,\ldots,x_r)\) be a tuple in a
unital \(C^*\)-algebra \(\mc A\). We say that \(\boldsymbol X_N\) converges strongly in
probability to \(\boldsymbol x\) if, for every \(D\geq1\) and every matrix-valued
noncommutative \(^*\)-polynomial
\[
    P\in M_D\otimes
      \mb C\langle z_1,\ldots,z_r,z_1^*,\ldots,z_r^*\rangle,
\]
one has
\begin{equation}\label{eq:intro-strong-convergence}
    \bigl\|P(\boldsymbol X_N,\boldsymbol X_N^*)\bigr\|
    \longrightarrow
    \bigl\|P(\boldsymbol x,\boldsymbol x^*)\bigr\|
\end{equation}
in probability. For deterministic tuples, convergence in probability is replaced by
ordinary convergence.
\end{definition}
 
Here a matrix-valued noncommutative \(^*\)-polynomial of degree at most \(s\) is a finite
sum
\[
    P=A_\emptyset\otimes 1+\sum_{j=1}^{s}\ \sum_{w_1,\ldots,w_j}
    A_{w_1\cdots w_j}\otimes w_1\cdots w_j,
\]
where each letter \(w_i\) ranges over \(\{z_1,\ldots,z_r,z_1^*,\ldots,z_r^*\}\) and the
coefficients \(A_{w_1\cdots w_j}\in M_D\) are \(D\times D\) complex matrices. Substituting
bounded operators in a \(C^*\)-algebra for the letters yields a bounded operator.
 
Strong convergence was first established for independent GUE matrices by Haagerup and
Thorbj{\o}rnsen~\cite{HaagerupThorbjornsen2005}, extended to Haar unitary matrices by
Collins and Male~\cite{CollinsMale2014}, and to random permutation matrices, restricted
to the orthogonal complement of the constant vector, by Bordenave and
Collins~\cite{BordenaveCollins2019}; see~\cite{vanhandel2026} for a recent survey. The
property we use is that, for self-adjoint polynomials, strong convergence is equivalent
to Hausdorff convergence of spectra. Allowing arbitrary matrix coefficients is essential
here, since it gives access to the full spectral information of the operators that
appear below.
 
\subsubsection*{A general principle: channels with infinite-dimensional input}
 
We next formulate channels with possibly infinite-dimensional input and
finite-dimensional output. Let \((\mc M,\tau)\) be a finite von Neumann algebra with a
faithful normal tracial state \(\tau\). We write \(\mc M^*\) for its Banach dual, that
is, the space of bounded linear functionals on \(\mc M\), and we use the full
\(C^*\)-algebraic state space
\begin{equation}\label{eq:intro-state-space}
    \St(\mc M)
    :=
    \{\omega\in\mc M^*: \omega\geq0,\ \omega(1)=1\},
\end{equation}
which includes singular states.
 
\begin{definition}\label{def:intro-vna-channel}
A channel with input algebra \(\mc M\) and output dimension \(k\) is specified in the
Heisenberg picture by a unital completely positive map
\begin{equation}\label{eq:intro-ucp-channel}
    \Gamma:M_k\longrightarrow\mc M.
\end{equation}
Its Schr\"odinger channel \(\Phi_\Gamma:\St(\mc M)\to\Dens_k\) is determined by
\begin{equation}\label{eq:intro-duality}
    \Tr\!\left(A\Phi_\Gamma(\omega)\right)
    =
    \omega(\Gamma(A)),
    \qquad A\in M_k,
    \quad \omega\in\St(\mc M),
\end{equation}
where \(\Dens_k=\{\rho\in M_k:\rho\geq0,\ \Tr\rho=1\}\).
\end{definition}
 
The one-copy output body of \(\Phi_\Gamma\) is
\[
    \mc O_\Gamma
    :=
    \Phi_\Gamma(\St(\mc M))
    \subseteq\Dens_k .
\]
For self-adjoint \(A\in M_k\), its support function satisfies
\begin{equation}\label{eq:intro-support-function}
    h_\Gamma(A)
    :=
    \max_{\rho\in\mc O_\Gamma}\Tr(A\rho)
    = \sup_{\omega \in \St(\mc M)} \omega(\Gamma(A))
    =
    \lambda_{\max}(\Gamma(A)),
\end{equation}
where \(\lambda_{\max}(\cdot)\) denotes the largest element of the spectrum of a bounded
self-adjoint operator. Moreover, the Gibbs variational principle gives
\begin{equation}\label{eq:intro-moe-dual}
    \MOE(\Phi_\Gamma)
    =
    \inf_{A=A^*}
    \left\{
        \log\Tr(e^A)-\lambda_{\max}(\Gamma(A))
    \right\}.
\end{equation}
Thus, even for infinite-dimensional input, the minimum output entropy depends only on the
largest eigenvalues of the Heisenberg observables \(\Gamma(A)\).
 
This observation is what makes strong convergence applicable. Suppose that the unital
completely positive maps are given by matrix-valued polynomials,
\[
    \Gamma_N(A)
    =
    P_A(\boldsymbol X_N,\boldsymbol X_N^*),
    \qquad
    \Gamma_\infty(A)
    =
    P_A(\boldsymbol x,\boldsymbol x^*),
\]
where the tuples \(\boldsymbol X_N\) converge strongly to \(\boldsymbol x\) and \(P_A\)
depends linearly on \(A\). Strong convergence, together with compactness of the unit
ball of self-adjoint observables, then gives
\begin{equation}\label{eq:intro-uniform-support}
    \sup_{A=A^*,\ \|A\|\leq1}
    \left|
        h_{\Gamma_N}(A)-h_{\Gamma_\infty}(A)
    \right|
    \longrightarrow0.
\end{equation}
Equivalently, the output bodies converge in Hausdorff distance, and therefore
\begin{equation}\label{eq:intro-moe-convergence}
    \MOE(\Phi_{\Gamma_N})
    \longrightarrow
    \MOE(\Phi_{\Gamma_\infty}).
\end{equation}
This principle yields the finite-dimensional approximation of the one-copy entropy.
 
The two-copy side is governed by a different identity. Let \(\overline{\mc M}\) be the
conjugate algebra and define the conjugate channel by
\(\overline{\Gamma}(A):=\overline{\Gamma(\overline A)}\). On
\(\mc M\otimes_{\max}\overline{\mc M}\)\footnote{It is essential to use the maximal
tensor product here: the underlying von Neumann algebra is a non-injective algebra of
type \(\mathrm{II}_1\), and the diagonal state does not define a state on the spatial
tensor product. See \Cref{subsec:von-neumann-basics} for the relevant definitions.},
consider the diagonal state, the infinite-dimensional analogue of the maximally
entangled state,
\begin{equation}\label{eq:intro-diagonal-state}
    \omega_\delta(a\otimes\overline b)
    :=
    \tau(ab^*),\qquad a,b\in \mc M.
\end{equation}
Writing
\(\rho_\delta^\Gamma:=(\Phi_\Gamma\otimes_{\max}\overline{\Phi_\Gamma})(\omega_\delta)\),
one has the Gram-matrix identity
\begin{equation}\label{eq:intro-gram-identity}
    \langle i,k|
        \rho_\delta^\Gamma
    |j,\ell\rangle
    =
    \tau\!\left(
        \Gamma(E_{ij})^*
        \Gamma(E_{k\ell})
    \right).
\end{equation}
Consequently,
\begin{equation}\label{eq:intro-max-criterion}
    H(\rho_\delta^\Gamma)
    <
    2\MOE(\Phi_\Gamma)
\end{equation}
is a sufficient criterion for
\(\MOE(\Phi_\Gamma\otimes_{\max}\overline{\Phi_\Gamma})<2\MOE(\Phi_\Gamma)\). Our
finite-dimensional construction achieves two things:
\begin{itemize}
    \item strong convergence ensures the approximation of the one-copy entropy lower bound;
    \item the finite-dimensional construction reproduces the two-copy state
    \(\rho_\delta^\Gamma\) exactly.
\end{itemize}
 
\subsubsection*{The free-compression channel}
 
We focus on the random-subspace model studied
in~\cite{AubrunSzarekWerner2011,BelinschiCollinsNechita2016} and its free limit. Fix
integers \(k,m\geq2\), choose \(q\in\{1,\ldots,m-1\}\), and put \(t:=q/m\). Let
\begin{equation}\label{eq:intro-free-product}
    (\mc A_{k,m},\tau)
    :=
    (M_k,\tr_k)*(M_m,\tr_m)
\end{equation}
be the reduced tracial free product of the two matrix algebras, and let
\(p\in M_m\subseteq\mc A_{k,m}\) be a rank-\(q\) projection. Thus \(\tau(p)=t\), and \(p\)
is free from the first copy of \(M_k\). The formal definitions of freeness and reduced
free products are recalled in \Cref{subsec:free-independence}. The limiting input
algebra is the corner \(\mc M_{k,t}:=p\mc A_{k,m}p\), equipped with the normalized trace
\(\tau_p(x):=t^{-1}\tau(x)\) for \(x\in\mc M_{k,t}\). The free-compression Heisenberg
channel is
\begin{equation}\label{eq:intro-free-channel}
    \Gamma_{k,t}:M_k\longrightarrow\mc M_{k,t},
    \qquad
    \Gamma_{k,t}(A)=pAp,
\end{equation}
and we write \(\Phi_{k,t}:=\Phi_{\Gamma_{k,t}}\) for the associated Schr\"odinger channel.
With the squared fidelity \(F\) introduced in \Cref{subsec:intro-construction}, the
output body of this channel is~\cite{BelinschiCollinsNechita2016}
\begin{equation}\label{eq:intro-fidelity-body}
    \mc K_{k,t}
    :=
    \left\{
        \rho\in\Dens_k:
        F\!\left(\rho,\frac{I_k}{k}\right)
        \geq1-t
    \right\}.
\end{equation}
Accordingly, we write \(H^*(k,t):=\min_{\rho\in\mc K_{k,t}}H(\rho)\) and define the
limiting entropy defect
\begin{equation}\label{eq:intro-free-gap}
    \Delta_{k,t}
    :=
    2H^*(k,t)
    -
    H\!\left(
        t\psi_k^+
        +(1-t)\frac{I_{k^2}}{k^2}
    \right).
\end{equation}
Belinschi, Collins, and Nechita showed that \(\Delta_{k,t}>0\) for suitable parameters,
already for output dimension \(k=183\), and that the defect can approach \(\log2\) as
\(k\to\infty\)~\cite{BelinschiCollinsNechita2016}. A small output dimension keeps the
number \(k^2m^2\) of permutations in our construction moderate. Since the inequality is
strict and the expressions depend continuously on \(t\), one may choose a rational
parameter \(t=q/m\).
 
\subsubsection*{Finite-dimensional permutation channels}
 
We now return to the permutation channels \(\Phi_N^{\boldsymbol\sigma}\) constructed in
\Cref{subsec:intro-construction}, with parameters \(k,m,q\) such that
\(\Delta_{k,t}>0\) for \(t=q/m\). It remains to state precisely what is required of the
tuple \(\boldsymbol\sigma_N\). Let \(h_{\mc K_{k,t}}\) be the support function, defined as
in~\eqref{eq:intro-support-function}, of the free-compression output
body~\eqref{eq:intro-fidelity-body}. We call the \(k^2m^2\)-tuple \(\boldsymbol\sigma_N\)
\emph{\(\varepsilon\)-admissible} if
\begin{equation}\label{eq:intro-admissible-permutations}
  \sup_{\substack{A=A^*\in M_k\\ \|A\|\leq1}}
  \left\{
    \lambda_{\max}\!\left(\Gamma_N^{\boldsymbol\sigma}(A)\right)
    -h_{\mc K_{k,t}}(A)
  \right\}
  \leq\varepsilon.
\end{equation}
This one-sided condition is exactly what is needed to keep every output of
\(\Phi_N^{\boldsymbol\sigma}\) within trace-norm distance \(\varepsilon\) of
\(\mc K_{k,t}\). Consequently, by the Fannes--Audenaert continuity
bound~\cite{Audenaert2007},
\begin{equation}\label{eq:intro-admissible-moe}
  \MOE(\Phi_N^{\boldsymbol\sigma})
  \geq H^*(k,t)-\omega_k(\varepsilon),
  \qquad
  \omega_k(\varepsilon)
  :=h_2(\varepsilon/2)+\frac{\varepsilon}{2}\log(k-1),
\end{equation}
where \(h_2\) is the binary entropy. On the other hand, the phase
system~\eqref{eq:intro-phase-system-short} gives, for every permutation tuple,
\begin{equation}\label{eq:intro-exact-bell-short}
  (\Phi_N^{\boldsymbol\sigma}\otimes
   \overline{\Phi_N^{\boldsymbol\sigma}})(\psi_{r_N}^+)
  =t\psi_k^++(1-t)\frac{I_{k^2}}{k^2},
\end{equation}
where \(r_N:=\rank P_N\) and \(\psi_{r_N}^+\) is the maximally entangled input state.
Hence, if \(2\omega_k(\varepsilon)<\Delta_{k,t}\), every \(\varepsilon\)-admissible tuple
yields
\[
  \MOE(\Phi_N^{\boldsymbol\sigma}\otimes
  \overline{\Phi_N^{\boldsymbol\sigma}})
  <2\MOE(\Phi_N^{\boldsymbol\sigma}).
\]
 
Finally, after conjugation by \(D_N\), the operators \(\Gamma_N^{\boldsymbol\sigma}(A)\)
are evaluations of a fixed degree-two matrix-valued \(^*\)-polynomial in the tuple
\((S_{\xi,N})_{\xi\in\Theta}\). Thus Bordenave--Collins strong convergence implies that
independent uniform permutations are \(\varepsilon\)-admissible with probability tending
to one, while the uniform spectral theorem of O'Donnell--Wu produces arbitrarily large
admissible tuples deterministically in polynomial time for every fixed \(k,m,q\) and
\(\varepsilon\)~\cite{BordenaveCollins2019,ODonnellWu2020}. The free-group realization,
the phase identity~\eqref{eq:intro-exact-bell-short}, and the verification
of~\eqref{eq:intro-admissible-permutations} are given in the later construction sections.

\subsection{Discussions and organization}\label{subsec:intro-contributions}
We point out several directions that are open. 
\begin{itemize}
    \item \textbf{Superadditivity of classical capacity}. First of all, our result provides a route for finite dimensional approximation of channels on von Neumann algebra (noninjective type $\mathrm{II}_1$), and we can actually apply this trick to the channels studied in \cite{CollinsYoun2022} to find finite dimensional approximations. This will provide some other explicit constructions for one-shot additivity violation of minimum output entropy. However, it remains unclear whether the regularized minimum output entropy is additive, equivalently, whether we have 
\begin{equation}
    C(\Phi \otimes \Psi) = C(\Phi) + C(\Psi)
\end{equation}
where $C$ is the classical capacity of quantum channels~\eqref{eq:intro-classical-capacity}.
\item \textbf{Explicit closed-form formula.} The algorithm discovered in \cite{ODonnellWu2020} produces a deterministic search. One particular achievement of that algorithm is the strong convergence property needed in our analysis. However, it remains not clear, whether a closed form or algebraic formula exists. Existing constructions of explicit expander graphs using Kazhdan's property (T) and group representations~\cite{margulis1988explicit, lubotzky1988ramanujan} seem to be a fundamental challenge to establish strong convergence, as remarked in \cite[before Remark 6.1]{vanhandel2026}. 
\item \textbf{Other mechanism beyond complex conjugate pair.} Our constructions, including all the other known random constructions use the same mechanism given by tensor product of complex conjugate pair. It would be very interesting to known other mechanism for additivity violation of minimum output entropy. 
\item \textbf{Optimal violation}. It seems that several random constructions, including the deterministic construction in this paper, provide a violation $\log 2$ for $2 H_{\min}(\Phi) - H_{\min}(\Phi \otimes \overline \Phi) $ and it is achieved when both the input and output dimensions tend to infinity. It would be interesting to show a larger violation using a different strong convergence phenomenon. 
\item \textbf{Explicit estimate of the size of the system, and implementation.} The construction does not provide a practical size to implement the algorithm. One may need to use the special property of the model, to get a sharp estimates on the size of the system. Moreover, it remains a challenge to efficiently implement the channel using quantum circuits. 
\end{itemize}
The remainder of the paper is organized as follows. First, we review the operator-algebraic preliminaries,
the channel duality, and the strong-convergence transfer principle.  We then
analyze the free-compression channel and construct its explicit free-group
realization and auxiliary phase system.  Finally, we apply
Bordenave--Collins random permutations and the O'Donnell--Wu deterministic
lift algorithm to get nonadditivity via random and deterministic permutations, and using Chen, Garza-Vargas, Tropp and van Handel to get a numerically explicit estimate on the size of the system.

\section{Preliminaries}\label{sec:preliminaries}
 
Throughout, all $C^*$-algebras are complex and unital, all states are
normalized, and all tensor products are taken over $\mb C$.  When a von
Neumann algebra appears in a $C^*$-tensor product, it is regarded as its
underlying $C^*$-algebra.  We write $\tr_n:=\frac1n\Tr$
for the normalized trace on $M_n$.
 
\subsection{\texorpdfstring{$C^*$}{C-star}-algebras, von Neumann algebras,
  states, and channels}
\label{subsec:cstar-basics}\label{subsec:von-neumann-basics}
 
\paragraph{$C^*$-algebras and states.}
A complex algebra $\mc A$ is a $*$-algebra if it carries an involution
$a\mapsto a^*$ satisfying
\[
  (ab)^*=b^*a^*,
  \qquad
  (\alpha a+\beta b)^*=\overline\alpha a^*+\overline\beta b^*.
\]
 
\begin{definition}[$C^*$-algebra]\label{def:cstar-algebra}
A \emph{$C^*$-algebra} is a Banach $*$-algebra $\mc A$ whose norm satisfies
\begin{equation}\label{eq:cstar-identity}
  \|a^*a\|=\|a\|^2,
  \qquad a\in\mc A.
\end{equation}
\end{definition}
 
The basic examples are $M_n$ and $\mc B(\mc H)$ with their operator norms.
By the Gelfand--Naimark theorem, every $C^*$-algebra admits a faithful
$*$-representation on a Hilbert space; see, for example,
\cite{BrownOzawa2008}.
 
An element $a\in\mc A$ is \emph{positive}, written $a\geq0$, if
$a=b^*b$ for some $b\in\mc A$.  A linear functional
$\omega:\mc A\to\mb C$ is positive if $\omega(a)\geq0$ whenever $a\geq0$.
The state space of $\mc A$ is
\begin{equation}\label{eq:cstar-state-space}
  \St(\mc A)
  :=
  \{\omega\in\mc A^*: \omega\geq0,\ \omega(\one)=1\},
\end{equation}
where $\mc A^*$ is the Banach dual.  In the weak-$*$ topology,
$\St(\mc A)$ is compact and convex.  For $\mc A=M_n$, every state has the
form
\[
  \omega_\rho(A)=\Tr(\rho A)
\]
for a unique density matrix
\[
  \rho\in\Dens_n
  :=
  \{\rho\in M_n:\rho\geq0,\ \Tr\rho=1\}.
\]
 
\paragraph{Von Neumann algebras and traces.}
Let $\mc H$ be a Hilbert space.  For $\mc S\subseteq\mc B(\mc H)$, its
commutant is
\[
  \mc S'
  :=
  \{T\in\mc B(\mc H):TS=ST\text{ for every }S\in\mc S\}.
\]
 
\begin{definition}[von Neumann algebra]\label{def:von-neumann-algebra}
A \emph{von Neumann algebra} on $\mc H$ is a unital $*$-subalgebra
$\mc M\subseteq\mc B(\mc H)$ satisfying
\[
  \mc M=\mc M''.
\]
Equivalently, $\mc M$ is closed in the weak operator topology, or in the
strong operator topology.
\end{definition}
 
A positive functional $\omega$ on $\mc M$ is \emph{normal} if
\[
  \omega\!\left(\sup_i x_i\right)=\sup_i\omega(x_i)
\]
for every bounded increasing net of positive elements.  The normal
functionals form the predual $\mc M_*$, and
\[
  \St_{\rm n}(\mc M):=\St(\mc M)\cap\mc M_*
\]
is the normal state space.  In this paper we use the full state space
$\St(\mc M)$, which may also contain singular states.  This convention
matters because, for $x=x^*\in\mc M$,
\begin{equation}\label{eq:state-spectral-edge}
  \max_{\omega\in\St(\mc M)}\omega(x)
  =
  \max\spec(x).
\end{equation}
 
A state $\tau$ on $\mc M$ is \emph{tracial} if
\[
  \tau(xy)=\tau(yx),
  \qquad x,y\in\mc M.
\]
We call $(\mc M,\tau)$ a \emph{tracial von Neumann algebra} if $\tau$ is
faithful, normal, and tracial.  Its GNS Hilbert space
$L_2(\mc M,\tau)$ is the completion of $\mc M$ with respect to
\[
  \ip{x}{y}_{2,\tau}:=\tau(x^*y),
  \qquad
  \|x\|_{2,\tau}:=\tau(x^*x)^{1/2}.
\]
Left multiplication gives a faithful normal representation
\[
  \ell_\tau(a)x:=ax,
  \qquad a,x\in\mc M,
\]
and the vector $\Omega_\tau:=\one$ satisfies
\[
  \tau(a)=\ip{\Omega_\tau}{\ell_\tau(a)\Omega_\tau}_{2,\tau}.
\]
We generally identify $\mc M$ with its left-multiplication representation.
 
For the finite-dimensional identifications used later, put
\[
  \mc H_n:=L_2(M_n,\tr_n),
  \qquad
  \ell_n(A)X:=AX.
\]
If $\tr_{\mc B(\mc H_n)}:=n^{-2}\Tr_{\mc H_n}$ is the normalized trace on
$\mc B(\mc H_n)$, then
\begin{equation}\label{eq:left-representation-trace}
  \tr_{\mc B(\mc H_n)}(\ell_n(A))=\tr_n(A),
  \qquad A\in M_n.
\end{equation}
Indeed, under vectorization, $\ell_n(A)$ is unitarily equivalent to
$A\otimes I_n$.
 
\paragraph{Spatial and maximal tensor products.}
Let $\mc A$ and $\mc B$ be $C^*$-algebras, and let
$\mc A\odot\mc B$ denote their algebraic tensor product.  Choose faithful
representations
\[
  \pi:\mc A\to\mc B(\mc H),
  \qquad
  \rho:\mc B\to\mc B(\mc K).
\]
The \emph{minimal}, or \emph{spatial}, $C^*$-norm is
\[
  \|x\|_{\min}
  :=
  \|(\pi\odot\rho)(x)\|_{\mc B(\mc H\otimes\mc K)},
  \qquad x\in\mc A\odot\mc B.
\]
It does not depend on the choice of faithful representations, and the
completion is written $\mc A\otimes_{\min}\mc B$.
 
The \emph{maximal} $C^*$-norm is
\begin{equation}\label{eq:maximal-tensor-norm}
  \left\|\sum_{j=1}^r a_j\otimes b_j\right\|_{\max}
  :=
  \sup_{(\pi_0,\rho_0)}
  \left\|
    \sum_{j=1}^r \pi_0(a_j)\rho_0(b_j)
  \right\|,
\end{equation}
where the supremum runs over pairs of $*$-representations of $\mc A$ and
$\mc B$ on a common Hilbert space with commuting ranges.  The completion is
denoted by $\mc A\otimes_{\max}\mc B$.  Since
$\|x\|_{\min}\leq\|x\|_{\max}$, the identity on the algebraic tensor product
extends to a canonical quotient map
\begin{equation}\label{eq:max-to-min-quotient}
  q_{\max,\min}:
  \mc A\otimes_{\max}\mc B
  \longrightarrow
  \mc A\otimes_{\min}\mc B.
\end{equation}
If either factor is nuclear, this map is an isomorphism.  In particular, the
minimal and maximal tensor products agree when one factor is a matrix
algebra; see \cite{BrownOzawa2008}.
 
If $\mc M\subseteq\mc B(\mc H)$ and
$\mc N\subseteq\mc B(\mc K)$ are von Neumann algebras, their
\emph{spatial von Neumann tensor product} is
\begin{equation}\label{eq:spatial-von-neumann-tensor}
  \mc M\overline\otimes\mc N
  :=
  (\mc M\odot\mc N)''
  \subseteq
  \mc B(\mc H\otimes\mc K).
\end{equation}
The norm closure of $\mc M\odot\mc N$ in this representation is
$\mc M\otimes_{\min}\mc N$.  If $(\mc M,\tau)$ and $(\mc N,\varphi)$ are
tracial, then $\tau\otimes\varphi$ extends uniquely to a faithful normal
tracial state on $\mc M\overline\otimes\mc N$.
 
Thus $\mc M\overline\otimes\mc N$ and $\mc M\otimes_{\max}\mc N$ are
different completions of the same algebraic tensor product.  The spatial
tensor product is generated by operators acting on separate tensor factors of
a Hilbert space, whereas the maximal tensor product is universal for
commuting representations on a single Hilbert space.  Every state on the
minimal tensor product pulls back through \eqref{eq:max-to-min-quotient} to a
state on the maximal tensor product, but a state on the maximal tensor
product need not factor through the minimal one.  The distinction disappears
in finite dimensions.
 
\paragraph{Completely positive maps and channels.}
For a linear map $\Gamma:\mc A\to\mc B$, let
\[
  \Gamma^{(r)}
  :=
  \id_{M_r}\otimes\Gamma:
  M_r(\mc A)\to M_r(\mc B).
\]
The map $\Gamma$ is \emph{completely positive} if $\Gamma^{(r)}$ is positive
for every $r\geq1$, and it is \emph{unital completely positive}, abbreviated
UCP, if it is completely positive and $\Gamma(\one)=\one$.
 
\begin{definition}[Channel with finite-dimensional output]
\label{def:prelim-vna-channel}
Let $\mc M$ be a von Neumann algebra.  A channel with input algebra $\mc M$
and output dimension $k$ is specified in the Heisenberg picture by a UCP map
\begin{equation}\label{eq:prelim-ucp-channel}
  \Gamma:M_k\longrightarrow\mc M.
\end{equation}
Its Schr\"odinger-picture channel
$\Phi_\Gamma:\St(\mc M)\to\Dens_k$ is the unique affine map satisfying
\begin{equation}\label{eq:prelim-channel-duality}
  \Tr\!\left(A\Phi_\Gamma(\omega)\right)
  =
  \omega(\Gamma(A)),
  \qquad
  A\in M_k,
  \quad
  \omega\in\St(\mc M).
\end{equation}
\end{definition}
 
Indeed, $\omega\circ\Gamma$ is a state on $M_k$ and is therefore represented
by a unique density matrix.  When $\mc M=M_n$, this is the usual
finite-dimensional Schr\"odinger channel, and $\Gamma$ is its
Hilbert--Schmidt adjoint.
 
The one-copy output body is
\begin{equation}\label{eq:prelim-output-body}
  \mc O_\Gamma
  :=
  \Phi_\Gamma(\St(\mc M))
  \subseteq\Dens_k.
\end{equation}
It is compact and convex.  For $A=A^*\in M_k$, its support function is
\begin{equation}\label{eq:prelim-support-function}
  h_\Gamma(A)
  :=
  \max_{\rho\in\mc O_\Gamma}\Tr(A\rho)
  =
  \max_{\omega\in\St(\mc M)}\omega(\Gamma(A))
  =
  \max\spec(\Gamma(A)),
\end{equation}
where the last equality follows from \eqref{eq:state-spectral-edge}.  Hence,
by the Gibbs variational principle,
\begin{equation}\label{eq:prelim-moe-dual}
  \MOE(\Phi_\Gamma)
  =
  \min_{\rho\in\mc O_\Gamma}H(\rho)
  =
  \inf_{A=A^*}
  \left\{
    \log\Tr(e^A)-h_\Gamma(A)
  \right\}.
\end{equation}
 
Let
\[
  \Gamma_i:M_{k_i}\to\mc M_i,
  \qquad i=1,2,
\]
be UCP maps.  By functoriality of the minimal and maximal tensor products for
completely positive maps, there are UCP maps
\begin{align}
  \Gamma_1\otimes_{\max}\Gamma_2:
  M_{k_1k_2}
  &\longrightarrow
  \mc M_1\otimes_{\max}\mc M_2,
  \label{eq:max-product-heisenberg}\\
  \Gamma_1\overline\otimes\Gamma_2:
  M_{k_1k_2}
  &\longrightarrow
  \mc M_1\overline\otimes\mc M_2,
  \label{eq:spatial-product-heisenberg}
\end{align}
characterized on elementary tensors by
\[
  (\Gamma_1\otimes\Gamma_2)(A\otimes B)
  =
  \Gamma_1(A)\otimes\Gamma_2(B).
\]
We define the corresponding product channels by duality:
\begin{align}
  \Phi_{\Gamma_1}\otimes_{\max}\Phi_{\Gamma_2}
  &:=
  \Phi_{\Gamma_1\otimes_{\max}\Gamma_2},
  \label{eq:max-product-schrodinger}\\
  \Phi_{\Gamma_1}\otimes_{\rm sp}\Phi_{\Gamma_2}
  &:=
  \Phi_{\Gamma_1\overline\otimes\Gamma_2}.
  \label{eq:spatial-product-schrodinger}
\end{align}
The first channel acts on $\St(\mc M_1\otimes_{\max}\mc M_2)$, and the
second on $\St(\mc M_1\overline\otimes\mc M_2)$.  These definitions include
entangled input states; they are not merely the pointwise tensor product of
the two affine maps on product states.  In finite dimensions, both reduce to
the usual tensor-product channel.
 
\paragraph{Conjugate algebras and the diagonal state.}
For a $C^*$-algebra $\mc A$, let $\overline{\mc A}$ denote its conjugate
$C^*$-algebra, and write $a\mapsto\overline a$ for the canonical
conjugate-linear $*$-isomorphism.  The map $\overline a\mapsto a^*$
identifies $\overline{\mc A}$ linearly and $*$-isomorphically with the
opposite algebra $\mc A^{\rm op}$.  For matrices, the bar denotes entrywise
complex conjugation in the standard basis.  If $\Gamma:M_k\to\mc M$ is UCP,
its conjugate is the UCP map
\begin{equation}\label{eq:prelim-conjugate-channel}
  \overline\Gamma:M_k\longrightarrow\overline{\mc M},
  \qquad
  \overline\Gamma(A)
  :=
  \overline{\Gamma(\overline A)}.
\end{equation}
We write $\overline{\Phi_\Gamma}:=\Phi_{\overline\Gamma}$.
 
Let $(\mc M,\tau)$ be tracial.  Besides the left action $\ell_\tau$, the
conjugate algebra acts on $L_2(\mc M,\tau)$ by
\begin{equation}\label{eq:right-conjugate-action}
  r_\tau(\overline b)x:=xb^*,
  \qquad b,x\in\mc M.
\end{equation}
The representations $\ell_\tau(\mc M)$ and $r_\tau(\overline{\mc M})$
commute, so by the universal property of the maximal tensor product they
induce a $*$-representation
\[
  \pi_\delta:
  \mc M\otimes_{\max}\overline{\mc M}
  \longrightarrow
  \mc B(L_2(\mc M,\tau)),
  \qquad
  \pi_\delta(a\otimes\overline b)
  =
  \ell_\tau(a)r_\tau(\overline b).
\]
The vector state of $\Omega_\tau=\one$ is the \emph{diagonal state}
\begin{equation}\label{eq:prelim-diagonal-state}
  \omega_\delta(z)
  :=
  \ip{\Omega_\tau}{\pi_\delta(z)\Omega_\tau}_{2,\tau},
  \qquad
  z\in\mc M\otimes_{\max}\overline{\mc M}.
\end{equation}
On elementary tensors,
\begin{equation}\label{eq:prelim-diagonal-state-elementary}
  \omega_\delta(a\otimes\overline b)
  =
  \tau(ab^*),
  \qquad a,b\in\mc M.
\end{equation}
The maximal tensor product is the natural domain here: the construction uses
two commuting representations on the same Hilbert space and requires no
spatial factorization.  If $\mc M=M_n$ and $\tau=\tr_n$, then under
vectorization $\Omega_\tau$ corresponds to
\[
  |\Omega_n\rangle
  =
  \frac1{\sqrt n}\sum_{j=1}^n e_j\otimes e_j,
\]
so $\omega_\delta$ is the usual maximally entangled vector state.
 
Finally, for a UCP map $\Gamma:M_k\to\mc M$, define
\begin{equation}\label{eq:prelim-diagonal-output}
  \rho_\delta^\Gamma
  :=
  (\Phi_\Gamma\otimes_{\max}\overline{\Phi_\Gamma})(\omega_\delta)
  \in\Dens_{k^2}.
\end{equation}
If $(E_{ij})_{i,j=1}^k$ are the standard matrix units, then
\begin{equation}\label{eq:prelim-gram-identity}
  \langle i,k|\rho_\delta^\Gamma|j,\ell\rangle
  =
  \tau\!\left(
    \Gamma(E_{ij})^*\Gamma(E_{k\ell})
  \right).
\end{equation}
Thus $\rho_\delta^\Gamma$ is the Gram matrix, with respect to $\tau$, of the
family $(\Gamma(E_{ij}))_{i,j=1}^k$.  This is the identity used in the
introduction to compute the two-copy output independently of the one-copy
spectral approximation.
 
\subsection{Free groups, free independence, and group von Neumann algebras}
\label{subsec:free-independence}\label{subsec:group-von-neumann}
 
\paragraph{Free groups and their regular representations.}
Let $\Theta$ be a finite set.
 
\begin{definition}[Free group]\label{def:free-group}
The \emph{free group on the generating set $\Theta$} is
\[
  \mb F_\Theta
  :=
  \langle g_\theta:\theta\in\Theta\rangle,
\]
that is, the group generated by the symbols $g_\theta$ subject only to the
relations required by the group axioms.  Its elements are the empty word $e$
and the reduced words
\begin{equation}\label{eq:reduced-free-group-word}
  g_{\theta_1}^{\varepsilon_1}\cdots
  g_{\theta_r}^{\varepsilon_r},
  \qquad
  \varepsilon_j\in\{1,-1\},
\end{equation}
in which adjacent letters never cancel:
$(\theta_{j+1},\varepsilon_{j+1})\neq
(\theta_j,-\varepsilon_j)$.
\end{definition}
 
More generally, let $G$ be a countable discrete group with identity $e$.
Its left regular representation
\[
  \lambda_G:G\to\mc U(\ell_2(G)),
  \qquad
  \lambda_G(g)\delta_h=\delta_{gh},
\]
satisfies
\[
  \lambda_G(g)\lambda_G(h)=\lambda_G(gh),
  \qquad
  \lambda_G(g)^*=\lambda_G(g^{-1}).
\]
The reduced group $C^*$-algebra and the group von Neumann algebra are,
respectively,
\begin{align}
  C_r^*(G)
  &:={\overline{\operatorname{span}}}^{\|\cdot\|}
    \{\lambda_G(g):g\in G\},
  \label{eq:reduced-group-cstar}\\
  \mc L(G)
  &:={\{\lambda_G(g):g\in G\}}''
  ={\overline{C_r^*(G)}}^{\rm WOT}.
  \label{eq:group-von-neumann-algebra}
\end{align}
 
\begin{definition}[Group von Neumann algebra]\label{def:group-von-neumann}
The von Neumann algebra $\mc L(G)$ in
\eqref{eq:group-von-neumann-algebra} is called the \emph{group von Neumann
algebra} of $G$.
\end{definition}
 
The vector $\delta_e$ defines the canonical trace
\begin{equation}\label{eq:canonical-group-trace}
  \tau_G(x)
  :=
  \ip{\delta_e}{x\delta_e}_{\ell_2(G)},
  \qquad x\in\mc L(G).
\end{equation}
For every finite sum,
\begin{equation}\label{eq:canonical-trace-coefficient}
  \tau_G\!\left(\sum_{g\in F}c_g\lambda_G(g)\right)
  =c_e.
\end{equation}
In particular,
\begin{equation}\label{eq:trace-group-unitary}
  \tau_G(\lambda_G(g))
  =
  \begin{cases}
    1,&g=e,\\
    0,&g\neq e.
  \end{cases}
\end{equation}
The trace $\tau_G$ is faithful, normal, and tracial, and the map
$\lambda_G(g)\mapsto\delta_g$ extends to a canonical unitary
\[
  L_2(\mc L(G),\tau_G)
  \cong
  \ell_2(G).
\]
If $|\Theta|\geq2$, then $\mb F_\Theta$ has infinite conjugacy classes
(it is ICC), and $\mc L(\mb F_\Theta)$ is a type $\mathrm{II}_1$ factor.
 
\paragraph{Free independence.}
Let $(\mc M,\tau)$ be a tracial von Neumann algebra.  For a unital
$*$-subalgebra $\mc A\subseteq\mc M$, write
\[
  \mc A^\circ
  :=
  \{a\in\mc A:\tau(a)=0\}.
\]
 
\begin{definition}[Free independence]\label{def:freeness}
Unital $*$-subalgebras $\mc A_i\subseteq\mc M$, $i\in I$, are
\emph{freely independent}, or simply \emph{free}, if
\begin{equation}\label{eq:freeness-definition}
  \tau(x_1x_2\cdots x_r)=0
\end{equation}
whenever $x_j\in\mc A_{i_j}^\circ$ and
$i_j\neq i_{j+1}$ for every $1\leq j<r$.  A family of elements is free if
the unital $*$-algebras that they generate are free.
\end{definition}
 
A unitary $u\in\mc M$ is a \emph{Haar unitary} if
\[
  \tau(u^n)=0,
  \qquad n\in\mb Z\setminus\{0\}.
\]
For the free group, put
\[
  u_\theta
  :=
  \lambda_{\mb F_\Theta}(g_\theta)
  \in\mc L(\mb F_\Theta).
\]
Each $u_\theta$ is a Haar unitary, and the family
$(u_\theta)_{\theta\in\Theta}$ is free.  Indeed, every nonempty reduced word
$w\in\mb F_\Theta$ is different from the identity, so
\begin{equation}\label{eq:reduced-word-zero-trace}
  \tau_{\mb F_\Theta}
  \bigl(\lambda_{\mb F_\Theta}(w)\bigr)=0.
\end{equation}
The moment identities that define freeness therefore reduce to the
uniqueness of reduced words in a free group.
 
\paragraph{Reduced tracial free products.}
Given tracial von Neumann algebras $(\mc M_1,\tau_1)$ and
$(\mc M_2,\tau_2)$, their reduced tracial free product
\[
  (\mc M,\tau)
  =
  (\mc M_1,\tau_1)*(\mc M_2,\tau_2)
\]
is characterized, up to trace-preserving isomorphism, by trace-preserving
embeddings of $\mc M_1$ and $\mc M_2$ into $\mc M$ whose images are free and
generate $\mc M$ as a von Neumann algebra.  It can be constructed on the
free-product Hilbert space
\begin{equation}\label{eq:free-product-Hilbert-space}
  \mb C\Omega
  \oplus
  \bigoplus_{r\geq1}
  \bigoplus_{\substack{i_1,\ldots,i_r\in\{1,2\}\\
    i_j\neq i_{j+1}}}
  \mc H_{i_1}^\circ\otimes\cdots\otimes\mc H_{i_r}^\circ,
\end{equation}
where
\[
  \mc H_i^\circ
  :=
  L_2(\mc M_i,\tau_i)\ominus\mb C\one;
\]
see \cite{VoiculescuDykemaNica1992}.  The vacuum vector $\Omega$ induces the
free-product trace.
 
The group construction is compatible with free products: for countable
discrete groups $G_1$ and $G_2$,
\begin{equation}\label{eq:group-vn-free-product}
  (\mc L(G_1),\tau_{G_1})
  *
  (\mc L(G_2),\tau_{G_2})
  \cong
  (\mc L(G_1*G_2),\tau_{G_1*G_2}).
\end{equation}
Indeed, an alternating product of nonidentity elements from the two factors
is a nonidentity reduced word in $G_1*G_2$ and therefore has zero canonical
trace.
 
\paragraph{The free-compression notation.}
Fix integers $k,m\geq2$, choose $q\in\{1,\ldots,m-1\}$, and put
$t=q/m$.  We use the reduced tracial free product
\begin{equation}\label{eq:prelim-matrix-free-product}
  (\mc A_{k,m},\tau)
  :=
  (M_k,\tr_k)*(M_m,\tr_m).
\end{equation}
Let $p\in M_m\subseteq\mc A_{k,m}$ be a rank-$q$ projection.  Then
$\tau(p)=t$, the first copy of $M_k$ is free from the second copy of $M_m$
and hence from $p$, and
\begin{equation}\label{eq:prelim-free-compression-corner}
  \mc M_{k,t}:=p\mc A_{k,m}p,
  \qquad
  \tau_p(x):=t^{-1}\tau(x),
\end{equation}
is a tracial von Neumann algebra with unit $p$.  The compression map
\[
  M_k\longrightarrow\mc M_{k,t},
  \qquad
  A\longmapsto pAp,
\]
is UCP when the unit of the codomain is taken to be $p$.
 
For the permutation construction, we use the index set
\begin{equation}\label{eq:subspace-free-group-index-set}
  \Theta_{k,m}
  :=
  \mb Z_k^2\times\mb Z_m^2
\end{equation}
and the free group
\begin{equation}\label{eq:subspace-free-group}
  \mb F_{\Theta_{k,m}}
  =
  \left\langle
    g_{a,b,c,d}:
    (a,b,c,d)\in\mb Z_k^2\times\mb Z_m^2
  \right\rangle.
\end{equation}
We write
\[
  u_{a,b,c,d}
  :=
  \lambda_{\mb F_{\Theta_{k,m}}}(g_{a,b,c,d})
  \in\mc L(\mb F_{\Theta_{k,m}}).
\]
These free Haar unitaries are the limiting operators that the standard
representations of the permutation tuples approximate in the
finite-dimensional construction.  Operator norms of polynomials in these
unitaries are the same in $C_r^*(\mb F_{\Theta_{k,m}})$ and in its weak
closure $\mc L(\mb F_{\Theta_{k,m}})$, and the latter supplies the canonical
trace used in the free-probabilistic calculations.

\subsection{A concrete realization of matrix-algebra free products}
\label{subsec:matrix-free-product-realization}
 
We record a concrete realization of the reduced free product
$(M_k,\tr_k)*(M_m,\tr_m)$, which is used in the finite-dimensional
construction.  It is a vectorized form of the controlled-unitary construction
of Ching~\cite{Ching1973}.
 
For $n\geq2$, let $\omega_n=e^{2\pi\mathrm i/n}$ and define the Weyl
unitaries on the standard basis $(e_j)_{j\in\mb Z_n}$ by
\begin{equation}\label{eq:Weyl-unitary}
  X_ne_j=e_{j+1},
  \qquad
  Z_ne_j=\omega_n^j e_j,
  \qquad
  W_{a,b}^{(n)}:=X_n^aZ_n^b .
\end{equation}
They form an orthonormal basis of $L_2(M_n,\tr_n)$:
\begin{equation}\label{eq:Weyl-orthonormality}
  \tr_n\!\left((W_\alpha^{(n)})^*W_{\alpha'}^{(n)}\right)
  =\delta_{\alpha,\alpha'},
  \qquad
  \alpha,\alpha'\in\mb Z_n^2.
\end{equation}
Define the normalized vectorization
\begin{equation}\label{eq:Weyl-Bell-vectorization}
  \operatorname{vec}_n(A)
  :=\frac1{\sqrt n}\sum_{i,j\in\mb Z_n}A_{ij}e_i\otimes e_j.
\end{equation}
Then
\begin{equation}\label{eq:Weyl-Bell-vectorization-identities}
  \ip{\operatorname{vec}_n(A)}{\operatorname{vec}_n(B)}
  =\tr_n(A^*B),
  \qquad
  \operatorname{vec}_n(CA)
  =(C\otimes I_n)\operatorname{vec}_n(A).
\end{equation}
For $\alpha\in\mb Z_n^2$, put
\begin{equation}\label{eq:Weyl-Bell-single}
  \ket{\alpha}_n:=\operatorname{vec}_n(W_\alpha^{(n)}).
\end{equation}
Thus $(\ket{\alpha}_n)_{\alpha\in\mb Z_n^2}$ is an orthonormal basis of
$\mb C^n\otimes\mb C^n$.  With
$\Theta=\mb Z_k^2\times\mb Z_m^2$, write $\xi=(\alpha,\beta)$ and set
\begin{equation}\label{eq:Weyl-Bell}
  \ket{\xi}:=\ket{\alpha}_k\otimes\ket{\beta}_m.
\end{equation}
 
Let $(\mc V,\varphi)$ be a tracial von Neumann algebra containing a free
family of Haar unitaries $(v_{\alpha,\beta})_{(\alpha,\beta)\in\Theta}$.  On
\begin{equation}
  \mc N(v)
  :=
  \mc B(\mb C^k\otimes\mb C^k)
  \overline\otimes
  \mc B(\mb C^m\otimes\mb C^m)
  \overline\otimes\mc V,
\end{equation}
equipped with its product trace, define
\begin{equation}\label{eq:controlled-D}
  D(v)
  :=
  \sum_{\alpha\in\mb Z_k^2}
  \sum_{\beta\in\mb Z_m^2}
  \ketbra{\alpha}{\alpha}_k
  \otimes\ketbra{\beta}{\beta}_m
  \otimes v_{\alpha,\beta}.
\end{equation}
The summands have orthogonal supports, so $D(v)$ is unitary.  Define
\begin{align}
  \pi_k^v(A)
  &:=(A\otimes I_k)\otimes I_{m^2}\otimes\one,
  \label{eq:pi-k-definition}\\
  \pi_m^v(B)
  &:=D(v)\bigl[I_{k^2}\otimes(B\otimes I_m)\otimes\one\bigr]D(v)^*.
  \label{eq:pi-m-definition}
\end{align}
The following theorem is essentially due to Ching~\cite{Ching1973}.
 
\begin{theorem}[Controlled-unitary realization of the free product]
\label{thm:explicit-matrix-free-product}
The maps $\pi_k^v$ and $\pi_m^v$ are trace-preserving normal unital
$*$-embeddings.  Their ranges are freely independent, and
\begin{equation}\label{eq:generated-free-product-isomorphism}
  W^*\!\left(\pi_k^v(M_k),\pi_m^v(M_m)\right)
  \cong
  (M_k,\tr_k)*(M_m,\tr_m)
\end{equation}
through a trace-preserving isomorphism carrying the two canonical free
factors to the displayed ranges.
\end{theorem}
 
\begin{proof}
Only freeness needs verification.  For $A\in M_k$ and $B\in M_m$, put
\begin{equation}
  a_{\alpha',\alpha}(A)
  :=\,{}_k\!\bra{\alpha'}(A\otimes I_k)\ket{\alpha}_k,
  \qquad
  b_{\beta',\beta}(B)
  :=\,{}_m\!\bra{\beta'}(B\otimes I_m)\ket{\beta}_m.
\end{equation}
Because every Weyl matrix is unitary,
\begin{equation}\label{eq:Weyl-diagonal-coefficients}
  a_{\alpha,\alpha}(A)=\tr_k(A),
  \qquad
  b_{\beta,\beta}(B)=\tr_m(B).
\end{equation}
Hence a centered element of the first copy changes the $\alpha$-label, and a
centered element of the second copy changes the $\beta$-label.  Moreover,
\begin{equation}\label{eq:controlled-m-action}
\begin{aligned}
  \pi_m^v(B)
  \bigl(\ket{\alpha}_k\otimes\ket{\beta}_m\otimes\zeta\bigr)
  =\sum_{\beta'}b_{\beta',\beta}(B)
  \ket{\alpha}_k\otimes\ket{\beta'}_m
  \otimes v_{\alpha,\beta'}v_{\alpha,\beta}^*\zeta.
\end{aligned}
\end{equation}
Thus every nonzero transition produced by a centered element of the second
copy contributes a two-letter block
$v_{\alpha,\beta'}v_{\alpha,\beta}^*$ with $\beta'\neq\beta$.  In an
alternating centered product, consecutive blocks of this kind are separated
by a centered element of the first copy, which changes $\alpha$ while
preserving the intervening $\beta$-label.  Therefore no cancellation occurs,
either inside a block or between adjacent blocks, and every nonzero summand
in the matrix-coefficient expansion carries a nonempty reduced word in the
free Haar family.  Consequently, every such summand has trace zero.  If the
product contains no element of the second copy, it consists of a single
centered element of the first copy and is traceless directly.  This proves
freeness.  The trace identities follow from normalized matrix traces and
invariance under unitary conjugation, and the universal characterization of
the reduced tracial free product gives
\eqref{eq:generated-free-product-isomorphism}.
\end{proof}
 
 
\section{The finite-dimensional channel and its free limit}
\label{sec:finite-dimensional-channel}
 
Fix $k,m\geq2$, choose $q\in\{1,\ldots,m-1\}$, and put
\begin{equation}
  t:=\frac qm.
\end{equation}
This section defines the finite-dimensional projection and its channel,
identifies the infinite-dimensional counterpart with the free-compression
channel, and verifies the two-copy Bell-state output exactly.  The one-copy
approximation is deferred to \Cref{sec:single-channel-output-convergence}.
 
\subsection{The finite projection and its channel}
\label{subsec:finite-projection-channel}
\label{subsec:closed-form-R-xi}
\label{subsec:finite-subspace-channel}
 
Write
\begin{equation}
  \Theta=\mc A_k\times\mc B_m,
  \qquad
  \mc A_k:=\mb Z_k^2,
  \qquad
  \mc B_m:=\mb Z_m^2.
\end{equation}
For $\alpha=(a,b)\in\mc A_k$ and $\beta=(c,d)\in\mc B_m$, define
\begin{equation}\label{eq:alpha-beta-integer-labels}
  \kappa(\alpha):=a+kb,
  \qquad
  \mu(\beta):=c+md.
\end{equation}
Set
\begin{equation}\label{eq:explicit-ell}
  \ell:=(k^2-1)(m^2-1)+1,
\end{equation}
let $\omega_\ell=e^{2\pi\mathrm i/\ell}$, and let
$Z_\ell e_s=\omega_\ell^s e_s$ on $\mb C^\ell$.  The auxiliary phase
unitaries are
\begin{equation}\label{eq:R-xi-closed-form}
  R_{\alpha,\beta}
  :=Z_\ell^{\,\kappa(\alpha)\mu(\beta)},
  \qquad
  (\alpha,\beta)\in\Theta.
\end{equation}
 
Let $N\geq2$, and let
$\boldsymbol\sigma_N=(\sigma_{\xi,N})_{\xi\in\Theta}$ be an arbitrary tuple
of permutations of $[N]$.  Put
\begin{equation}
  \ket{+}_N:=\frac1{\sqrt N}\sum_{j=1}^N e_j,
  \qquad
  \mc K_N:=\ket{+}_N^\perp,
  \qquad
  S_{\xi,N}:=P_{\sigma_{\xi,N}}\big|_{\mc K_N}.
\end{equation}
Define
\begin{equation}\label{eq:finite-Hn-vectorized}
\begin{aligned}
  \mc H_N^{\rm fin}
  &:=(\mb C^k\otimes\mb C^k)
    \otimes(\mb C^m\otimes\mb C^m)
    \otimes\mc K_N\otimes\mb C^\ell\\
  &\cong\mb C^k\otimes\mc E_N,
  \qquad
  d_N:=\dim\mc E_N=km^2(N-1)\ell.
\end{aligned}
\end{equation}
For $\xi\in\Theta$, put
\begin{equation}\label{eq:T-xi-n}
  T_{\xi,N}:=S_{\xi,N}\otimes R_\xi,
\end{equation}
and define the controlled unitary
\begin{equation}\label{eq:Dn-vectorized}
  D_N^{\boldsymbol\sigma}
  :=\sum_{\xi\in\Theta}
    \ketbra{\xi}{\xi}\otimes T_{\xi,N}.
\end{equation}
On $\mc H_N^{\rm fin}$, let
\begin{align}
  \pi_{k,N}(A)
  &:=(A\otimes I_k)\otimes I_{m^2}
    \otimes I_{\mc K_N}\otimes I_\ell,
  \label{eq:finite-pi-vectorized}\\
  \rho_{m,N}(B)
  &:=I_{k^2}\otimes(B\otimes I_m)
    \otimes I_{\mc K_N}\otimes I_\ell,\notag\\
  \eta_{m,N}^{\boldsymbol\sigma}(B)
  &:=D_N^{\boldsymbol\sigma}\rho_{m,N}(B)
    (D_N^{\boldsymbol\sigma})^* .\notag
\end{align}
Choose the rank-$q$ projection
\begin{equation}\label{eq:Pt-explicit}
  p_q:=\sum_{j=0}^{q-1}\ketbra{e_j}{e_j}\in M_m.
\end{equation}
From now on, we take the projection $p$ in the abstract free-compression
model to be the copy of $p_q$ in the second free factor.  This entails no
loss of generality, since all rank-$q$ projections in $M_m$ are unitarily
conjugate.  Set
\begin{equation}\label{eq:Pn-definition-vectorized}
\begin{aligned}
  Q_N
  &:=I_{k^2}\otimes(p_q\otimes I_m)
    \otimes I_{\mc K_N}\otimes I_\ell,\\
  P_N^{\boldsymbol\sigma}
  &:=D_N^{\boldsymbol\sigma}Q_N(D_N^{\boldsymbol\sigma})^*
    =\eta_{m,N}^{\boldsymbol\sigma}(p_q).
\end{aligned}
\end{equation}
The projection has rank
\begin{equation}\label{eq:rn-rank}
  r_N:=\rank P_N^{\boldsymbol\sigma}
  =k^2qm(N-1)\ell
  =t k d_N.
\end{equation}
Let
\begin{equation}
  \mc M_N^{\boldsymbol\sigma}
  :=P_N^{\boldsymbol\sigma}
    \mc B(\mc H_N^{\rm fin})
    P_N^{\boldsymbol\sigma},
  \qquad
  \tau_N^{\boldsymbol\sigma}(X)
  :=\frac1{r_N}\Tr(X).
\end{equation}
Its unit is $P_N^{\boldsymbol\sigma}$.  The finite-dimensional Heisenberg
channel is
\begin{equation}\label{eq:finite-Heisenberg-channel}
  \Gamma_N^{\boldsymbol\sigma}:M_k\longrightarrow
  \mc M_N^{\boldsymbol\sigma},
  \qquad
  \Gamma_N^{\boldsymbol\sigma}(A)
  :=P_N^{\boldsymbol\sigma}\pi_{k,N}(A)P_N^{\boldsymbol\sigma}.
\end{equation}
It is UCP.  We define its Schr\"odinger-picture channel by duality:
\begin{equation}\label{eq:finite-channel-duality}
\begin{aligned}
  \Phi_N^{\boldsymbol\sigma}:
  \St(\mc M_N^{\boldsymbol\sigma})&\longrightarrow\Dens_k,\\
  \Tr\!\left(A\Phi_N^{\boldsymbol\sigma}(\omega)\right)
  &=\omega\!\left(\Gamma_N^{\boldsymbol\sigma}(A)\right),
  \qquad A\in M_k.
\end{aligned}
\end{equation}
Equivalently, choose an isometry
$J_N:\mb C^{r_N}\to\mb C^k\otimes\mc E_N$ with
$J_NJ_N^*=P_N^{\boldsymbol\sigma}$.  Under the induced identification
$\mc M_N^{\boldsymbol\sigma}\cong M_{r_N}$,
\begin{equation}\label{eq:Phi-n-definition}
  \Phi_N^{\boldsymbol\sigma}(X)
  =\Tr_{\mc E_N}(J_NXJ_N^*),
  \qquad X\in M_{r_N}.
\end{equation}
Thus \eqref{eq:finite-channel-duality}, rather than the choice of $J_N$, is
the intrinsic definition of the channel.  When the permutation tuple is
fixed, we suppress the superscript $\boldsymbol\sigma$ and write simply $D_N$,
$P_N$, $\mc M_N$, $\Gamma_N$, and $\Phi_N$.
 
\subsection{Comparison with the infinite-dimensional free-compression channel}
\label{subsec:finite-infinite-comparison}
 
Let $u_\xi=\lambda_{\mb F_\Theta}(g_\xi)$ be the canonical free Haar
unitaries from \Cref{subsec:group-von-neumann}, and put
\begin{equation}
  v_\xi:=u_\xi\otimes R_\xi
  \in\mc L(\mb F_\Theta)\overline\otimes M_\ell.
\end{equation}
The assignment $g_\xi\mapsto R_\xi$ extends to a unitary representation
$\rho_R:\mb F_\Theta\to\mc U(\ell)$.  Hence, for every word
$w\in\mb F_\Theta$,
\begin{equation}
  v(w)=\lambda_{\mb F_\Theta}(w)\otimes\rho_R(w),
\end{equation}
and its trace is zero unless $w=e$.  Consequently, $(v_\xi)_{\xi\in\Theta}$
is a free family of Haar unitaries with the same joint $*$-distribution as
$(u_\xi)_\xi$.
 
We apply \Cref{thm:explicit-matrix-free-product} to this family.  In
\begin{equation}
  \mc N_\infty
  :=\mc B(\mb C^k\otimes\mb C^k)
    \overline\otimes
    \mc B(\mb C^m\otimes\mb C^m)
    \overline\otimes
    \mc L(\mb F_\Theta)
    \overline\otimes M_\ell,
\end{equation}
with product trace $\varphi_\infty$, define
\begin{equation}\label{eq:controlled-D-vectorized-free-limit}
  D_\infty
  :=\sum_{\xi\in\Theta}
    \ketbra{\xi}{\xi}\otimes u_\xi\otimes R_\xi.
\end{equation}
Let
\begin{align}
  \pi_{k,\infty}(A)
  &:=(A\otimes I_k)\otimes I_{m^2}\otimes\one\otimes I_\ell,\\
  \pi_{m,\infty}(B)
  &:=D_\infty
    \bigl[I_{k^2}\otimes(B\otimes I_m)\otimes\one\otimes I_\ell\bigr]
    D_\infty^*,
\end{align}
and put
\begin{equation}
  \mc A_\infty
  :=W^*\!\left(\pi_{k,\infty}(M_k),\pi_{m,\infty}(M_m)\right).
\end{equation}
The infinite-dimensional projection and corner are
\begin{align}
  Q_\infty
  &:=I_{k^2}\otimes(p_q\otimes I_m)\otimes\one\otimes I_\ell,
  \label{eq:section3-vectorized-Q-infinity}\\
  P_\infty
  &:=\pi_{m,\infty}(p_q)
  =D_\infty Q_\infty D_\infty^*,
  \label{eq:section3-vectorized-p-infinity}\\
  \mc M_\infty
  &:=P_\infty\mc A_\infty P_\infty,
  \qquad
  \tau_\infty:=t^{-1}\varphi_\infty\big|_{\mc M_\infty}.
\end{align}
Define
\begin{equation}\label{eq:infinite-Heisenberg-channel}
  \Gamma_\infty:M_k\longrightarrow\mc M_\infty,
  \qquad
  \Gamma_\infty(A):=P_\infty\pi_{k,\infty}(A)P_\infty,
\end{equation}
and let $\Phi_\infty:=\Phi_{\Gamma_\infty}$ be its Schr\"odinger-picture
channel.
 
\begin{proposition}[Identification with free compression]
\label{prop:free-model-identification}
There is a normal, trace-preserving $*$-isomorphism
\begin{equation}
  \Psi:(\mc A_{k,m},\tau)\longrightarrow(\mc A_\infty,\varphi_\infty)
\end{equation}
that carries the canonical copies of $M_k$ and $M_m$ to
$\pi_{k,\infty}(M_k)$ and $\pi_{m,\infty}(M_m)$, respectively.  In
particular, $\Psi(p)=P_\infty$, and $\Psi$ restricts to an isomorphism
\begin{equation}
  (\mc M_{k,t},\tau_p)
  \cong
  (\mc M_\infty,\tau_\infty)
\end{equation}
that intertwines $\Gamma_{k,t}$ and $\Gamma_\infty$.  Consequently,
\begin{equation}\label{eq:free-model-output-identification}
  \mc O_{\Gamma_\infty}=\mc K_{k,t},
  \qquad
  \MOE(\Phi_\infty)=H^*(k,t).
\end{equation}
\end{proposition}
 
\begin{proof}
The first assertion is \Cref{thm:explicit-matrix-free-product}.  Since the
rank-$q$ projection $p$ lies in the second free factor, $\Psi$ carries it to
$\pi_{m,\infty}(p_q)=P_\infty$.  Restricting $\Psi$ to the corresponding
corners gives the intertwining of the two channels.  The last identities
follow from the description \eqref{eq:intro-fidelity-body} of the output
body of the free-compression channel.
\end{proof}
 
\begin{table}[ht]
\centering
\caption{The finite-dimensional construction and its free limit; here
$\rho_{k,t}=t\psi_k^++(1-t)I_{k^2}/k^2$.}
\label{tab:finite-free-comparison}
\renewcommand{\arraystretch}{1.35}
\begin{tabular}{@{}>{\raggedright\arraybackslash}p{3.4cm}
                >{\raggedright\arraybackslash}p{5.8cm}
                >{\raggedright\arraybackslash}p{5.8cm}@{}}
\toprule
Object & Finite dimension  & Free limit \\
\midrule
Unitaries
  & $S_{\xi,N}=P_{\sigma_{\xi,N}}\big|_{\mc K_N}$ on $\mc K_N=\ket{+}_N^\perp$
  & $u_\xi=\lambda_{\mb F_\Theta}(g_\xi)\in\mc L(\mb F_\Theta)$ \\
  \midrule
Ambient algebra
  & $\mc B(\mc H_N^{\rm fin})$, normalized trace $\varphi_N$
  & $\mc N_\infty$, product trace $\varphi_\infty$ \\
  \midrule
Controlled unitary
  & $D_N^{\boldsymbol\sigma}=\sum_\xi\ketbra{\xi}{\xi}\otimes S_{\xi,N}\otimes R_\xi$
  & $D_\infty=\sum_\xi\ketbra{\xi}{\xi}\otimes u_\xi\otimes R_\xi$ \\
  \midrule
Projection
  & $P_N^{\boldsymbol\sigma}=D_N^{\boldsymbol\sigma}Q_N(D_N^{\boldsymbol\sigma})^*$
  & $P_\infty=D_\infty Q_\infty D_\infty^*$ \\
  \midrule
Input algebra
  & $\mc M_N^{\boldsymbol\sigma}=P_N^{\boldsymbol\sigma}\mc B(\mc H_N^{\rm fin})P_N^{\boldsymbol\sigma}
     \cong M_{r_N}$, with trace $\tau_N^{\boldsymbol\sigma}=t^{-1}\varphi_N$
  & $\mc M_\infty=P_\infty\mc A_\infty P_\infty\cong\mc M_{k,t}$,
     with trace $\tau_\infty=t^{-1}\varphi_\infty$ \\
     \midrule
Heisenberg channel
  & $\Gamma_N^{\boldsymbol\sigma}(A)=P_N^{\boldsymbol\sigma}\pi_{k,N}(A)P_N^{\boldsymbol\sigma}$
  & $\Gamma_\infty(A)=P_\infty\pi_{k,\infty}(A)P_\infty$ \\
Schr\"odinger channel
  & $\Phi_N^{\boldsymbol\sigma}:\St(\mc M_N^{\boldsymbol\sigma})\to\Dens_k$, by duality
  & $\Phi_\infty:\St(\mc M_\infty)\to\Dens_k$, by duality \\
\midrule
One-copy output body
  & $\mc O_{\Gamma_N^{\boldsymbol\sigma}}$, close to $\mc K_{k,t}$ for suitable
     $\boldsymbol\sigma_N$ (\Cref{sec:single-channel-output-convergence})
  & $\mc O_{\Gamma_\infty}=\mc K_{k,t}$ (\Cref{prop:free-model-identification}) \\
Two-copy output
  & $(\Phi_N^{\boldsymbol\sigma}\otimes\overline{\Phi_N^{\boldsymbol\sigma}})(\psi_{r_N}^+)
     =\rho_{k,t}$ for every $\boldsymbol\sigma_N$
     (\Cref{thm:exact-isotropic-Bell-general-R})
  & $\rho_\delta^{\Gamma_\infty}=\rho_{k,t}$ \\
\bottomrule
\end{tabular}
\end{table}

\Cref{tab:finite-free-comparison} compares the finite-dimensional construction with its
free limit.  In each row, the free-limit entry is obtained from the finite one by the
formal replacement $S_{\xi,N}\mapsto u_\xi$, together with
$\mc B(\mc K_N)\mapsto\mc L(\mb F_\Theta)$.  
 
\subsection{The exact finite-dimensional Bell-state phenomenon}
\label{subsec:finite-Bell-phenomenon}
\label{sec:exact-Bell-output}
\label{subsec:exact-Bell-general-R}
 
The role of the auxiliary phases is captured by one elementary identity.
 
\begin{lemma}[Rectangle trace identity]
\label{lem:rectangle-trace-explicit-R}
For $\alpha_0,\alpha_1\in\mc A_k$ and
$\beta_0,\beta_1\in\mc B_m$,
\begin{equation}\label{eq:rectangle-trace-indicator}
\begin{aligned}
  \tr_\ell\!\left(
    R_{\alpha_1,\beta_0}
    R_{\alpha_1,\beta_1}^*
    R_{\alpha_0,\beta_1}
    R_{\alpha_0,\beta_0}^*
  \right) \quad = \quad
  \mathbf 1_{\{\alpha_0=\alpha_1\ \mathrm{or}\ \beta_0=\beta_1\}}.
\end{aligned}
\end{equation}
Equivalently,
\begin{equation}\label{eq:rectangular-trace-orthogonality}
  \tr_\ell\!\left(
    R_{\alpha_1,\beta_0}
    R_{\alpha_1,\beta_1}^*
    R_{\alpha_0,\beta_1}
    R_{\alpha_0,\beta_0}^*
  \right)=0
\end{equation}
whenever $\alpha_0\neq\alpha_1$ and $\beta_0\neq\beta_1$.
\end{lemma}
 
\begin{proof}
The product inside the trace equals
\begin{equation}
  Z_\ell^{\,[\kappa(\alpha_1)-\kappa(\alpha_0)]
                 [\mu(\beta_0)-\mu(\beta_1)]}.
\end{equation}
If one pair of labels agrees, the exponent is zero.  Otherwise it is a
nonzero integer of absolute value at most $(k^2-1)(m^2-1)=\ell-1$, hence it
is not divisible by $\ell$, and the normalized trace of the corresponding
nontrivial power of $Z_\ell$ vanishes.
\end{proof}
 
Let $\varphi_N$ denote the normalized trace on
$\mc B(\mc H_N^{\rm fin})$.  Since
$\rank P_N^{\boldsymbol\sigma}=t\dim\mc H_N^{\rm fin}$,
\begin{equation}\label{eq:corner-trace-relation}
  \tau_N^{\boldsymbol\sigma}(X)=t^{-1}\varphi_N(X),
  \qquad X\in\mc M_N^{\boldsymbol\sigma}.
\end{equation}
 
\begin{lemma}[Exact compressed moment]
\label{lem:exact-compressed-moment}
For every $A_1,A_2\in M_k$,
\begin{equation}\label{eq:compressed-moment-general-R}
\begin{aligned}
  &\varphi_N\!\left(
    \pi_{k,N}(A_1)P_N^{\boldsymbol\sigma}
    \pi_{k,N}(A_2)P_N^{\boldsymbol\sigma}
  \right) =
  t^2\tr_k(A_1A_2)
  +t(1-t)\tr_k(A_1)\tr_k(A_2).
\end{aligned}
\end{equation}
Equivalently,
\begin{equation}\label{eq:corner-channel-second-moment}
  \tau_N^{\boldsymbol\sigma}\!\left(
    \Gamma_N^{\boldsymbol\sigma}(A_1)
    \Gamma_N^{\boldsymbol\sigma}(A_2)
  \right)
  =t\tr_k(A_1A_2)
   +(1-t)\tr_k(A_1)\tr_k(A_2).
\end{equation}
Both identities hold for every $N$ and every permutation tuple
$\boldsymbol\sigma_N$.
\end{lemma}
 
\begin{proof}
As in the proof of \Cref{thm:explicit-matrix-free-product}, for
$A\in M_k$ and $B\in M_m$ let
\begin{equation}
  a_{\alpha',\alpha}(A)
  :=\,{}_k\!\bra{\alpha'}(A\otimes I_k)\ket{\alpha}_k,
  \qquad
  b_{\beta',\beta}(B)
  :=\,{}_m\!\bra{\beta'}(B\otimes I_m)\ket{\beta}_m.
\end{equation}
The Weyl basis gives
\begin{align}
  a_{\alpha,\alpha}(A)&=\tr_k(A),
  &
  \sum_{\alpha_0,\alpha_1}
  a_{\alpha_0,\alpha_1}(A_1)
  a_{\alpha_1,\alpha_0}(A_2)
  &=k^2\tr_k(A_1A_2),
  \label{eq:Weyl-sums-k}\\
  b_{\beta,\beta}(B)&=\tr_m(B),
  &
  \sum_{\beta_0,\beta_1}
  b_{\beta_0,\beta_1}(B_1)
  b_{\beta_1,\beta_0}(B_2)
  &=m^2\tr_m(B_1B_2).
  \label{eq:Weyl-sums-m}
\end{align}
Moreover, with the normalized trace on $\mc K_N\otimes\mb C^\ell$,
\begin{equation}
\begin{aligned}
  &\tr\!\left(
    T_{(\alpha_1,\beta_0),N}
    T_{(\alpha_1,\beta_1),N}^*
    T_{(\alpha_0,\beta_1),N}
    T_{(\alpha_0,\beta_0),N}^*
  \right)\\
  &\qquad=
  \mathbf 1_{\{\alpha_0=\alpha_1\ \mathrm{or}\ \beta_0=\beta_1\}}.
\end{aligned}
\end{equation}
Indeed, the product is the identity when either equality holds; otherwise
its $M_\ell$-trace vanishes by \eqref{eq:rectangular-trace-orthogonality},
whatever the permutation factor is.  Expanding in the Weyl--Bell basis and
applying inclusion--exclusion therefore gives, for arbitrary
$B_1,B_2\in M_m$,
\begin{equation}\label{eq:exact-fourth-moment-general-R}
\begin{aligned}
  &\varphi_N\!\left(
    \pi_{k,N}(A_1)\eta_{m,N}^{\boldsymbol\sigma}(B_1)
    \pi_{k,N}(A_2)\eta_{m,N}^{\boldsymbol\sigma}(B_2)
  \right)\\
  &=\tr_k(A_1A_2)\tr_m(B_1)\tr_m(B_2)\\
  &\quad+\tr_k(A_1)\tr_k(A_2)\tr_m(B_1B_2)\\
  &\quad-\tr_k(A_1)\tr_k(A_2)\tr_m(B_1)\tr_m(B_2).
\end{aligned}
\end{equation}
Taking $B_1=B_2=p_q$ and using $p_q^2=p_q$ and $\tr_m(p_q)=t$ yields
\eqref{eq:compressed-moment-general-R}.  Equation
\eqref{eq:corner-channel-second-moment} follows from
\eqref{eq:corner-trace-relation} and cyclicity of the trace.
\end{proof}
 
The following theorem is an exact, finite-dimensional analogue of the
Bell-state phenomenon, which Collins and Nechita established in the
large-dimension limit~\cite{CollinsNechita2010}.
 
\begin{theorem}[Exact Bell-state phenomenon]
\label{thm:exact-isotropic-Bell-general-R}
For every $N\geq2$ and every permutation tuple
$\boldsymbol\sigma_N$,
\begin{equation}\label{eq:exact-isotropic-Bell-general-R}
  (\Phi_N^{\boldsymbol\sigma}\otimes
   \overline{\Phi_N^{\boldsymbol\sigma}})(\psi_{r_N}^+)
  =t\psi_k^++(1-t)\frac{I_{k^2}}{k^2}.
\end{equation}
Thus the finite-dimensional channel exactly reproduces the diagonal-state
output of the free-compression channel.
\end{theorem}
 
\begin{proof}
Under the identification $\mc M_N^{\boldsymbol\sigma}\cong M_{r_N}$ induced
by $J_N$, the diagonal state associated with $\tau_N^{\boldsymbol\sigma}$ is
the maximally entangled state $\psi_{r_N}^+$.  Let $\rho_N$ denote the
left-hand side of \eqref{eq:exact-isotropic-Bell-general-R}.  By the
Gram-matrix identity \eqref{eq:prelim-gram-identity}, for
$1\leq i,j,a,b\leq k$,
\begin{align}
  \langle i,a|\rho_N|j,b\rangle
  &=\tau_N^{\boldsymbol\sigma}\!\left(
    \Gamma_N^{\boldsymbol\sigma}(E_{ij})^*
    \Gamma_N^{\boldsymbol\sigma}(E_{ab})
  \right)\\
  &=t\tr_k(E_{ji}E_{ab})
    +(1-t)\tr_k(E_{ji})\tr_k(E_{ab})\\
  &=\frac tk\delta_{ia}\delta_{jb}
    +\frac{1-t}{k^2}\delta_{ij}\delta_{ab},
\end{align}
where the second line uses \eqref{eq:corner-channel-second-moment}.  These
are precisely the matrix elements of $t\psi_k^++(1-t)I_{k^2}/k^2$.
\end{proof}
 
Consequently,
\begin{equation}
  \MOE\!\left(
    \Phi_N^{\boldsymbol\sigma}\otimes
    \overline{\Phi_N^{\boldsymbol\sigma}}
  \right)
  \leq
  H\!\left(t\psi_k^++(1-t)\frac{I_{k^2}}{k^2}\right)
\end{equation}
for every permutation tuple.  The two-copy side of the argument is therefore
exact; only the one-copy approximation remains, and it is treated in
\Cref{sec:single-channel-output-convergence}.
 

\section{The one-copy approximation and finite-dimensional nonadditivity}
\label{sec:single-channel-output-convergence}
 
Fix $k,m\geq2$, $q\in\{1,\ldots,m-1\}$, and $t=q/m$.  Let
$(R_\xi)_{\xi\in\Theta}\subseteq\mc U(\ell)$ satisfy the rectangle trace
identity of \Cref{sec:exact-Bell-output}.  Given
\[
  \boldsymbol\sigma_N
  :=(\sigma_{\xi,N})_{\xi\in\Theta}\in S_N^\Theta,
\]
let
\begin{equation}\label{eq:admissibility-finite-channel}
  \Phi_N^{\boldsymbol\sigma}:M_{r_N}\longrightarrow M_k,
  \qquad
  r_N=k^2qm\ell(N-1),
\end{equation}
be the channel of \Cref{sec:finite-dimensional-channel}, and denote its
Heisenberg adjoint by
\begin{equation}\label{eq:admissibility-Heisenberg-map}
  \Gamma_N^{\boldsymbol\sigma}:M_k\longrightarrow M_{r_N}.
\end{equation}
By the preceding section, the Bell-state calculation is exact for every
choice of the permutations:
\begin{equation}\label{eq:admissibility-exact-Bell}
  \left(
    \Phi_N^{\boldsymbol\sigma}
    \otimes
    \overline{\Phi_N^{\boldsymbol\sigma}}
  \right)(\psi_{r_N}^+)
  =t\psi_k^++(1-t)\frac{I_{k^2}}{k^2}.
\end{equation}
Consequently, the two-copy upper bound needs no approximation, and the only
remaining task is to compare the one-copy output body with the limiting body
$\mc K_{k,t}$.
 
We measure the relevant one-sided error of the tuple by
\begin{equation}\label{eq:admissibility-error}
  \operatorname{err}_N(\boldsymbol\sigma_N)
  :=
  \sup_{\substack{A=A^*\in M_k\\ \|A\|\leq1}}
  \left\{
    \lambda_{\max}\!\left(\Gamma_N^{\boldsymbol\sigma}(A)\right)
    -h_{\mc K_{k,t}}(A)
  \right\}.
\end{equation}
This is the quantity bounded in the introduction: condition
\eqref{eq:intro-admissible-permutations} states exactly that
$\operatorname{err}_N(\boldsymbol\sigma_N)\leq\varepsilon$.
 
\subsection{Support polynomials and a criterion for nonadditivity}
\label{subsec:admissibility-support-polynomial}
 
The support function is a top spectral edge, whereas strong convergence is
formulated in terms of operator norms.  The following elementary shift
relates the two.
 
\begin{lemma}[Support functions as shifted operator norms]
\label{lem:support-shifted-norm}
Let $\Gamma:M_k\to\mathcal A$ be a unital completely positive map into a
unital $C^*$-algebra, with output body $\mathcal O_\Gamma$.  If
$A=A^*\in M_k$ and $\|A\|\leq1$, then
\begin{equation}\label{eq:support-shifted-norm}
  h_\Gamma(A)
  =\max\spec\!\left(\Gamma(A)\right)
  =\bigl\|1_{\mathcal A}+\Gamma(A)\bigr\|-1.
\end{equation}
\end{lemma}
 
\begin{proof}
By duality,
\[
  h_\Gamma(A)
  =\sup_{\omega\in\St(\mathcal A)}\omega(\Gamma(A))
  =\max\spec(\Gamma(A)).
\]
A unital positive map is contractive on self-adjoint elements, so
$\spec(\Gamma(A))\subseteq[-1,1]$.  Hence
$1_{\mathcal A}+\Gamma(A)\geq0$, and its norm is
$1+\max\spec(\Gamma(A))$.
\end{proof}
 
We now encode these spectral edges by degree-two polynomials.  Moving the
permutation register to the last tensor position, set
\begin{equation}\label{eq:admissibility-coefficient-space}
  \mc E_0
  :=
  (\mb C^k\otimes\mb C^k)
  \otimes
  (\mb C^m\otimes\mb C^m)
  \otimes\mb C^\ell.
\end{equation}
Let $p_q\in M_m$ be the rank-$q$ projection used in the construction, and put
\begin{equation}\label{eq:admissibility-Q-G-r0}
  Q:=I_{k^2}\otimes(p_q\otimes I_m)\otimes I_\ell,
  \qquad
  \mc G:=\operatorname{Ran}Q,
  \qquad
  r_0:=\dim\mc G=k^2qm\ell.
\end{equation}
For $\xi\in\Theta$, write $\Pi_\xi=|\xi\rangle\langle\xi|$ and define
\begin{equation}\label{eq:admissibility-Vxi}
  V_\xi:=\Pi_\xi\otimes R_\xi\in\mc B(\mc E_0).
\end{equation}
For $A\in M_k$, let
\begin{equation}\label{eq:admissibility-Ahat}
  \widehat A
  :=((A\otimes I_k)\otimes I_{m^2})\otimes I_\ell
  \in\mc B(\mc E_0)
\end{equation}
and
\begin{equation}\label{eq:admissibility-Cxi-eta}
  C_{\xi,\eta}(A)
  :=\left.QV_\xi^*\widehat A V_\eta Q\right|_{\mc G}
  \in\mc B(\mc G)\cong M_{r_0}.
\end{equation}
 
Let
\begin{equation}\label{eq:admissibility-free-group}
  \mb F_\Theta=\langle z_\xi:\xi\in\Theta\rangle
\end{equation}
be the free group on $\Theta$, and write $Z_w$ for the basis element of the
group algebra associated with $w\in\mb F_\Theta$; in particular,
$Z_\xi:=Z_{z_\xi}$.  For $A=A^*\in M_k$, set
\begin{equation}\label{eq:admissibility-Aw}
  A_w(A)
  :=
  \sum_{\substack{\xi,\eta\in\Theta:\\
       \operatorname{red}(z_\xi^{-1}z_\eta)=w}}
  C_{\xi,\eta}(A),
  \qquad |w|\leq2,
\end{equation}
and define the self-adjoint reduced polynomial
\begin{equation}\label{eq:admissibility-pA}
\begin{aligned}
  \mathfrak p_A(Z)
  &:=\sum_{|w|\leq2}A_w(A)\otimes Z_w\\
  &=\sum_{\xi,\eta\in\Theta}
    C_{\xi,\eta}(A)\otimes Z_\xi^{-1}Z_\eta
  \in M_{r_0}\odot\mb C[\mb F_\Theta].
\end{aligned}
\end{equation}
 
For $\sigma\in S_N$, denote its standard representation by
\begin{equation}\label{eq:admissibility-standard-representation}
  U_\sigma
  :=P_\sigma\big|_{(\mb C^N)^\circ},
  \qquad
  (\mb C^N)^\circ
  :=\left\{x\in\mb C^N:\sum_{j=1}^N x_j=0\right\},
\end{equation}
and put
\begin{equation}\label{eq:admissibility-finite-free-tuples}
  \boldsymbol U_N
  :=(U_{\sigma_{\xi,N}})_{\xi\in\Theta}.
\end{equation}
Let
\begin{equation}\label{eq:admissibility-free-regular-tuple}
  L_\xi:=\lambda_{\mb F_\Theta}(z_\xi),
  \qquad
  \boldsymbol L:=(L_\xi)_{\xi\in\Theta}.
\end{equation}
 
\begin{lemma}[Support-polynomial representation]
\label{lem:admissibility-polynomial-representation}
For every $A=A^*\in M_k$,
\begin{align}
  \lambda_{\max}\!\left(\Gamma_N^{\boldsymbol\sigma}(A)\right)
  &=\lambda_{\max}\!\left(\mathfrak p_A(\boldsymbol U_N)\right),
  \label{eq:admissibility-finite-edge}\\
  h_{\mc K_{k,t}}(A)
  &=\lambda_{\max}\!\left(\mathfrak p_A(\boldsymbol L)\right).
  \label{eq:admissibility-free-edge}
\end{align}
Moreover, with
\begin{equation}\label{eq:admissibility-coefficient-constants}
  s:=|\Theta|=k^2m^2,
  \qquad
  R_0:=s\sqrt\ell,
\end{equation}
one has
\begin{equation}\label{eq:admissibility-coefficient-bound}
  \max_{|w|\leq2}\|A_w(A)\|_F\leq R_0
  \qquad
  \text{whenever }\|A\|\leq1.
\end{equation}
\end{lemma}
 
\begin{proof}
After reordering tensor factors, the controlled unitary is
\[
  D_N^{\boldsymbol\sigma}
  =\sum_{\xi\in\Theta}V_\xi\otimes U_{\sigma_{\xi,N}}.
\]
Conjugating the compression by $D_N^{\boldsymbol\sigma}$ gives
\[
  \Gamma_N^{\boldsymbol\sigma}(A)
  \cong
  \sum_{\xi,\eta\in\Theta}
  C_{\xi,\eta}(A)
  \otimes U_{\sigma_{\xi,N}}^*U_{\sigma_{\eta,N}}
  =\mathfrak p_A(\boldsymbol U_N),
\]
which proves \eqref{eq:admissibility-finite-edge}.  Replacing the standard
permutation representations by the left regular generators gives the
infinite-dimensional controlled unitary of
\Cref{sec:finite-dimensional-channel}, and the free-product identification of
\Cref{thm:explicit-matrix-free-product} identifies the resulting compression
with $A\mapsto pAp$.  Its output body is $\mc K_{k,t}$, so
\eqref{eq:admissibility-free-edge} follows from the support-function formula
\eqref{eq:prelim-support-function}.
 
Finally, $\Pi_\xi\widehat A\Pi_\eta$ has rank at most one and Frobenius norm
at most $\|A\|$, while $\|R_\xi^*R_\eta\|_F=\sqrt\ell$.  Hence
$\|C_{\xi,\eta}(A)\|_F\leq\sqrt\ell$.  A nontrivial reduced word
$z_\xi^{-1}z_\eta$ determines the ordered pair $(\xi,\eta)$ uniquely, and
only the identity coefficient receives $s$ summands.  The triangle
inequality therefore gives \eqref{eq:admissibility-coefficient-bound}.
\end{proof}
 
For the quantitative argument below, we use the positive shift
\begin{equation}\label{eq:admissibility-qA}
  \mathfrak q_A(Z):=\mathfrak p_A(Z)+I_{\mc G}\otimes Z_e.
\end{equation}
Let $\boldsymbol W=(W_\xi)_{\xi\in\Theta}$ be a unitary tuple on a
Hilbert space $\mathcal H$, and put
\[
  \mc D(\boldsymbol W):=\sum_{\xi\in\Theta}V_\xi\otimes W_\xi.
\]
Since the projections $\Pi_\xi$ are mutually orthogonal and sum to the
identity, $\mc D(\boldsymbol W)$ is unitary.  Moreover,
\begin{equation}\label{eq:admissibility-positive-shift-compression}
\begin{aligned}
  \mathfrak q_A(\boldsymbol W)
  ={}&(Q\otimes I_{\mathcal H})\mc D(\boldsymbol W)^*
  (\widehat{A+I_k}\otimes I_{\mathcal H})
  \mc D(\boldsymbol W)(Q\otimes I_{\mathcal H})
  \big|_{\mc G\otimes\mathcal H}.
\end{aligned}
\end{equation}
Consequently, for $A=A^*$ with $\|A\|\leq1$,
\begin{equation}\label{eq:admissibility-qA-bounds}
  0\leq\mathfrak q_A(\boldsymbol W)\leq2I,
  \qquad
  \|\mathfrak q_A\|_{M_{r_0}\otimes C^*(\mb F_\Theta)}\leq2,
\end{equation}
and
\begin{equation}\label{eq:admissibility-qA-edge-identities}
\begin{aligned}
  \|\mathfrak q_A(\boldsymbol U_N)\|
  &=1+\lambda_{\max}\!\left(\Gamma_N^{\boldsymbol\sigma}(A)\right),\\
  \|\mathfrak q_A(\boldsymbol L)\|
  &=1+h_{\mc K_{k,t}}(A).
\end{aligned}
\end{equation}
 
The same error has a direct geometric meaning.  For compact convex sets
$\mc C,\mc D\subseteq\Dens_k$, define
\begin{equation}\label{eq:admissibility-directed-distance}
  d_{\to,T}(\mc C,\mc D)
  :=\sup_{\rho\in\mc C}\inf_{\sigma\in\mc D}
  \frac12\|\rho-\sigma\|_1.
\end{equation}
Trace-norm/operator-norm duality gives
\begin{equation}\label{eq:admissibility-directed-support-duality}
  2d_{\to,T}(\mc C,\mc D)
  =
  \sup_{\substack{A=A^*\\ \|A\|\leq1}}
  \bigl(h_{\mc C}(A)-h_{\mc D}(A)\bigr).
\end{equation}
Thus, writing
$\mc O_N^{\boldsymbol\sigma}:=
\Phi_N^{\boldsymbol\sigma}(\Dens_{r_N})$,
\begin{equation}\label{eq:admissibility-error-distance-identity}
  \operatorname{err}_N(\boldsymbol\sigma_N)
  =2d_{\to,T}\!\left(
      \mc O_N^{\boldsymbol\sigma},\mc K_{k,t}
    \right).
\end{equation}
 
Put
\begin{equation}\label{eq:admissibility-entropy-modulus}
  \Omega_k(T):=h_2(T)+T\log(k-1),
  \qquad
  h_2(T):=-T\log T-(1-T)\log(1-T).
\end{equation}
 
\begin{proposition}[Small one-sided error implies nonadditivity]
\label{prop:admissibility-entropy-lower-bound}
If $\operatorname{err}_N(\boldsymbol\sigma_N)\leq\varepsilon$ and
$0\leq\varepsilon/2\leq1-1/k$, then
\begin{equation}\label{eq:admissibility-MOE-lower-bound}
  \MOE(\Phi_N^{\boldsymbol\sigma})
  \geq H^*(k,t)-\Omega_k(\varepsilon/2).
\end{equation}
In particular, if
\begin{equation}\label{eq:admissibility-gap-condition}
  2\Omega_k(\varepsilon/2)<\Delta_{k,t},
\end{equation}
then
\begin{equation}\label{eq:admissibility-nonadditivity-conclusion}
  \MOE\!\left(
    \Phi_N^{\boldsymbol\sigma}
    \otimes
    \overline{\Phi_N^{\boldsymbol\sigma}}
  \right)
  <2\MOE(\Phi_N^{\boldsymbol\sigma}).
\end{equation}
\end{proposition}
 
\begin{proof}
By \eqref{eq:admissibility-error-distance-identity}, every
$\rho\in\mc O_N^{\boldsymbol\sigma}$ admits some
$\sigma\in\mc K_{k,t}$ with
$\frac12\|\rho-\sigma\|_1\leq\varepsilon/2$.  The Fannes--Audenaert
inequality~\cite{Audenaert2007} gives
\[
  H(\rho)
  \geq H(\sigma)-\Omega_k(\varepsilon/2)
  \geq H^*(k,t)-\Omega_k(\varepsilon/2).
\]
Taking the infimum over $\rho$ proves
\eqref{eq:admissibility-MOE-lower-bound}.  On the other hand,
\eqref{eq:admissibility-exact-Bell} gives
\[
  \MOE\!\left(
    \Phi_N^{\boldsymbol\sigma}
    \otimes
    \overline{\Phi_N^{\boldsymbol\sigma}}
  \right)
  \leq
  H\!\left(t\psi_k^++(1-t)\frac{I_{k^2}}{k^2}\right).
\]
The conclusion follows from the definition of $\Delta_{k,t}$ in
\eqref{eq:intro-free-gap}.
\end{proof}

\subsection{Random and deterministic permutations}
\label{subsec:admissibility-random-deterministic}
 
The polynomial representation reduces the remaining argument to two known
results.  Bordenave--Collins handles random permutation tuples, while
O'Donnell--Wu supplies deterministic tuples that work uniformly for the whole
bounded class of polynomials.  We first recall the matrix-coefficient form of
the Bordenave--Collins theorem~\cite{BordenaveCollins2019}.
 
\begin{theorem}[Bordenave--Collins strong convergence]
\label{thm:admissibility-BC-specialized}
Fix $s\geq2$.  For each $N$, let
$\sigma_{1,N},\ldots,\sigma_{s,N}$ be independent uniform permutations of
$[N]$, and let
\[
  U_{j,N}:=P_{\sigma_{j,N}}\big|_{(\mb C^N)^\circ}.
\]
If $L_1,\ldots,L_s$ are the canonical left regular generators of
$\mb F_s$, then for every matrix-valued noncommutative $^*$-polynomial $P$,
\begin{equation}\label{eq:admissibility-BC-strong-convergence}
  \|P(\boldsymbol U_N,\boldsymbol U_N^*)\|
  \longrightarrow
  \|P(\boldsymbol L,\boldsymbol L^*)\|
\end{equation}
in probability.
\end{theorem}
 
\begin{theorem}[Random permutations]
\label{thm:admissibility-random}
If $\boldsymbol\sigma_N=(\sigma_{\xi,N})_{\xi\in\Theta}$ consists of
independent uniform permutations, then
\begin{equation}\label{eq:admissibility-random-uniform-convergence}
  \sup_{\substack{A=A^*\in M_k\\ \|A\|\leq1}}
  \left|
    \lambda_{\max}\!\left(\Gamma_N^{\boldsymbol\sigma}(A)\right)
    -h_{\mc K_{k,t}}(A)
  \right|
  \longrightarrow0
\end{equation}
in probability.  In particular, for every fixed $\varepsilon>0$,
\begin{equation}\label{eq:admissibility-random-probability}
  \Pr\!\left\{
    \operatorname{err}_N(\boldsymbol\sigma_N)\leq\varepsilon
  \right\}
  \longrightarrow1.
\end{equation}
\end{theorem}
 
\begin{proof}
Fix $A=A^*$ with $\|A\|\leq1$.  Applying
\Cref{thm:admissibility-BC-specialized} to $\mathfrak q_A$ and using
\eqref{eq:admissibility-qA-edge-identities} gives convergence of the upper
spectral edge.  For uniformity, cover the self-adjoint operator-norm unit
ball of $M_k$ by a finite $\delta$-net.  Each of the two support functions in
\eqref{eq:admissibility-random-uniform-convergence} is $1$-Lipschitz, so
their difference is $2$-Lipschitz.  Convergence on the finite net, followed
by $\delta\downarrow0$, proves the result.
\end{proof}
 
\begin{corollary}[Random permutation counterexamples]
\label{cor:admissibility-random-nonadditivity}
If $\Delta_{k,t}>0$, then independent uniform permutations satisfy
\begin{equation}
  \Pr\!\left\{
    \MOE\!\left(
      \Phi_N^{\boldsymbol\sigma}
      \otimes
      \overline{\Phi_N^{\boldsymbol\sigma}}
    \right)
    <2\MOE(\Phi_N^{\boldsymbol\sigma})
  \right\}
  \longrightarrow1.
\end{equation}
\end{corollary}
 
\begin{proof}
Choose $\varepsilon>0$ with $2\Omega_k(\varepsilon/2)<\Delta_{k,t}$ and
combine
\Cref{thm:admissibility-random,prop:admissibility-entropy-lower-bound}.
\end{proof}
 
For the deterministic construction, we use the following specialization of
the O'Donnell--Wu theorem.
 
\begin{theorem}[O'Donnell--Wu uniform polynomial lifts]
\label{thm:admissibility-OW}
Fix integers $s,r,D\geq1$ and constants $R,\delta>0$.  There is a
deterministic algorithm with the following property.  On input $M$, it runs
in time polynomial in $M$ and outputs an integer $N$ and permutations
$\sigma_1,\ldots,\sigma_s\in S_N$ such that
\begin{equation}\label{eq:admissibility-OW-size}
  M\leq N\leq M+o_{s,r,D,R,\delta}(M)
  \qquad(M\to\infty),
\end{equation}
and, simultaneously for every $1\leq r'\leq r$ and every self-adjoint
reduced polynomial
\begin{equation}
  \mathfrak p=\sum_{|w|\leq D}B_w\otimes Z_w,
  \qquad
  B_w\in M_{r'},
  \qquad
  \max_w\|B_w\|_F\leq R,
\end{equation}
one has
\begin{equation}\label{eq:admissibility-OW-spectrum}
  d_{\mathrm H}\!\left(
    \spec\mathfrak p(\boldsymbol U),
    \spec\mathfrak p(\boldsymbol L)
  \right)
  \leq\delta,
\end{equation}
where
$U_j=P_{\sigma_j}|_{(\mb C^N)^\circ}$ and
$\boldsymbol L$ is the left regular generating tuple of $\mb F_s$.
\end{theorem}
 
This is the specialization $d=0$, $e=s$ of
\cite[Theorem~10.13]{ODonnellWu2020}.
 
\begin{theorem}[Deterministic permutations]
\label{thm:admissibility-deterministic}
Apply \Cref{thm:admissibility-OW} with
\begin{equation}\label{eq:admissibility-OW-parameters}
  s=k^2m^2,
  \qquad
  r=r_0=k^2qm\ell,
  \qquad
  D=2,
  \qquad
  R=R_0=k^2m^2\sqrt\ell.
\end{equation}
For any tolerance $\delta>0$, the resulting tuple, indexed by $\Theta$ in a
fixed order, satisfies
\begin{equation}\label{eq:admissibility-OW-uniform-support}
  \sup_{\substack{A=A^*\in M_k\\ \|A\|\leq1}}
  \left|
    \lambda_{\max}\!\left(\Gamma_N^{\boldsymbol\sigma}(A)\right)
    -h_{\mc K_{k,t}}(A)
  \right|
  \leq\delta.
\end{equation}
In particular, $\operatorname{err}_N(\boldsymbol\sigma_N)\leq\delta$.
\end{theorem}
 
\begin{proof}
By \Cref{lem:admissibility-polynomial-representation}, the family
$\{\mathfrak p_A:A=A^*,\ \|A\|\leq1\}$ lies in the bounded class of
polynomials specified by \eqref{eq:admissibility-OW-parameters}.  The same
permutation lift therefore gives
\[
  d_{\mathrm H}\!\left(
    \spec\mathfrak p_A(\boldsymbol U_N),
    \spec\mathfrak p_A(\boldsymbol L)
  \right)
  \leq\delta
\]
simultaneously for every such $A$.  The Hausdorff distance between the
spectra controls the distance between their upper endpoints, and
\eqref{eq:admissibility-finite-edge}--\eqref{eq:admissibility-free-edge} then
give \eqref{eq:admissibility-OW-uniform-support}.
\end{proof}
 
\begin{corollary}[Polynomial-time deterministic counterexamples]
\label{cor:admissibility-deterministic-nonadditivity}
Assume $\Delta_{k,t}>0$, and choose $\delta_0>0$ such that
\[
  2\Omega_k(\delta_0/2)<\Delta_{k,t}.
\]
For every sufficiently large input $M$, the O'Donnell--Wu algorithm, run with
the parameters \eqref{eq:admissibility-OW-parameters} and tolerance
$\delta_0$, produces in time polynomial in $M$ a finite-dimensional channel
satisfying
\begin{equation}
  \MOE\!\left(
    \Phi_N^{\boldsymbol\sigma}
    \otimes
    \overline{\Phi_N^{\boldsymbol\sigma}}
  \right)
  <2\MOE(\Phi_N^{\boldsymbol\sigma}).
\end{equation}
\end{corollary}
 
\begin{proof}
Combine \Cref{thm:admissibility-deterministic} with
\Cref{prop:admissibility-entropy-lower-bound} at tolerance $\delta_0$.
\end{proof}
 
\begin{remark}[Asymptotic construction versus a numerical threshold]
\label{rem:admissibility-OW-no-numerical-threshold}
For fixed $s,r,D,R,\delta$, the O'Donnell--Wu algorithm runs in time
polynomial in the input size $M$.  The asymptotic relation $N=M+o(M)$,
however, gives no explicit threshold in terms of these fixed parameters.
The theorem therefore proves efficient asymptotic constructibility, but it
does not produce the numerical value of $N$ derived below.  That value comes
from a separate tail estimate for independent uniform permutations and
serves only as an explicit existence bound.
\end{remark}

\subsection{A quantitative existence bound}
\label{subsec:admissibility-quantitative}
 
To give a concrete scale to the construction, we next derive a sufficient
value of $N$.  The required one-sided estimate follows from the master
inequalities of Chen, Garza-Vargas, Tropp, and van Handel
\cite[Theorem~6.1 and Corollary~9.7]{ChenGarzaVargasTroppVanHandel2026}.
 
\begin{lemma}[Upper tail for a positive permutation polynomial]
\label{lem:admissibility-explicit-tail}
Let $U_{1,N},\ldots,U_{s,N}$ be the standard representations of independent
uniform permutations of $[N]$.  Let
$\mathfrak P\in M_r\odot\mb C[\mb F_s]$ be self-adjoint of degree at most
$D$, and assume that $\mathfrak P(\boldsymbol W)\geq0$ for every unitary
tuple $\boldsymbol W$.  Put
\begin{equation}
  K:=\|\mathfrak P\|_{M_r\otimes C^*(\mb F_s)},
  \qquad
  \rho:=\|\mathfrak P(\boldsymbol L)\|.
\end{equation}
Then, for every $\beta>0$,
\begin{equation}\label{eq:admissibility-explicit-tail}
  \Pr\!\left\{
    \|\mathfrak P(\boldsymbol U_N)\|>\rho+\beta
  \right\}
  \leq
  \frac{\pi r}{3N}
  \bigl[4D(1+\log s)\bigr]^8
  \left(\frac{150K}{\beta}\right)^9.
\end{equation}
\end{lemma}
 
\begin{proof}
The event is empty if $\rho+\beta\geq K$.  Otherwise, set
\[
  a:=\rho+\frac\beta2,
  \qquad
  b:=\rho+\beta,
  \qquad
  \delta:=\arccos(a/K)-\arccos(b/K).
\]
Then $\delta\geq\beta/(2K)$.  Define
\[
  S(u):=c\int_0^u t^9(1-t)^9\,dt,
  \qquad
  c:=\frac{19!}{(9!)^2},
  \qquad 0\leq u\leq1,
\]
and extend $S$ by $0$ and $1$ outside $[0,1]$.  Since
\[
  \|S^{(10)}\|_{L^1(0,1)}
  \leq
  c\sum_{j=0}^9
  \binom9j\frac{(9+j)!}{j!(j+1)}
  <75^9,
\]
the cutoff
\[
  \chi(x)
  :=
  \begin{cases}
    0,&x\leq a,\\
    \displaystyle
    S\!\left(
      \frac{\arccos(a/K)-\arccos(x/K)}{\delta}
    \right),&a<x<b,\\
    1,&x\geq b
  \end{cases}
\]
has cosine--Chebyshev coefficients $\chi(x)=\sum_{j\geq0}a_jT_j(x/K)$
satisfying
\begin{equation}\label{eq:admissibility-Chebyshev-sum}
  \sum_{j\geq1}j^8|a_j|
  <\frac\pi3\left(\frac{150K}{\beta}\right)^9.
\end{equation}
Indeed, $\theta\mapsto\chi(K\cos\theta)$ has two transition intervals of
length $\delta$, so its tenth derivative has $L^1$ norm at most
$2\delta^{-9}\|S^{(10)}\|_1$, and ten integrations by parts give
\eqref{eq:admissibility-Chebyshev-sum}.
 
Apply the master inequality with expansion order two to the Chebyshev series
above.  The limiting distribution and the first correction term are
supported in $[-\rho,\rho]$, whereas $\chi$ vanishes on a neighborhood of
that interval, so both terms vanish.  With the normalization of
\cite{ChenGarzaVargasTroppVanHandel2026}, in which the trace on the standard
representation is divided by $N$ rather than by $N-1$, we obtain
\[
  \mb E\,\frac{1}{rN}\Tr
  \chi\!\left(\mathfrak P(\boldsymbol U_N)\right)
  \leq
  \frac{\pi}{3N^2}
  \bigl[4D(1+\log s)\bigr]^8
  \left(\frac{150K}{\beta}\right)^9.
\]
If the event in \eqref{eq:admissibility-explicit-tail} occurs, positivity of
$\mathfrak P(\boldsymbol U_N)$ and the definition of $\chi$ imply
$\Tr\chi(\mathfrak P(\boldsymbol U_N))\geq1$.  Multiplying the last display by
$rN$ and applying Markov's inequality gives
\eqref{eq:admissibility-explicit-tail}.
\end{proof}
 
Applying the lemma to $\mathfrak P=\mathfrak q_A$ and using
\eqref{eq:admissibility-qA-bounds}--\eqref{eq:admissibility-qA-edge-identities}
gives, for each fixed $A=A^*$ with $\|A\|\leq1$,
\begin{equation}\label{eq:admissibility-fixed-A-tail}
\begin{aligned}
  \Pr\!\left\{
    \lambda_{\max}\!\left(\Gamma_N^{\boldsymbol\sigma}(A)\right)
    >h_{\mc K_{k,t}}(A)+\beta
  \right\}\leq
  \frac{\pi r_0}{3N}
  \bigl[8(1+\log s)\bigr]^8
  \left(\frac{300}{\beta}\right)^9.
\end{aligned}
\end{equation}

To make the estimate uniform without covering the scalar direction, let
\begin{equation}\label{eq:admissibility-quotient-space}
  \mathfrak X_k:=M_k^{\rm sa}/\mb R I_k,
  \qquad
  \|[A]\|_{\rm q}:=\inf_{c\in\mb R}\|A-cI_k\|,
  \qquad
  d_0:=\dim_{\mb R}\mathfrak X_k=k^2-1.
\end{equation}
The function
\[
  f_N([A])
  :=\lambda_{\max}\!\left(\Gamma_N^{\boldsymbol\sigma}(A)\right)
  -h_{\mc K_{k,t}}(A)
\]
is well defined and $2$-Lipschitz with respect to $\|\cdot\|_{\rm q}$.
For $d\geq3$, put
\begin{equation}\label{eq:admissibility-Rogers-density}
  \vartheta_d:=d\bigl(\log d+\log\log d+5\bigr).
\end{equation}
The standard net argument~\cite{Rogers1957}, combined with
the standard averaging argument, yields the following finite-dimensional
consequence: for every $\eta>0$, the quotient unit ball admits an $\eta$-net
with at most
\begin{equation}\label{eq:admissibility-Rogers-cover}
  \vartheta_{d_0}\left(1+\frac1\eta\right)^{d_0}
\end{equation}
centers, each lying in the quotient ball of radius $1+\eta$.
 
\begin{proposition}[An explicit sufficient size]
\label{prop:admissibility-explicit-N-condition}
Let $\eta,\beta>0$.  If
\begin{equation}\label{eq:admissibility-explicit-N-condition}
\begin{aligned}
  N>{}&
  \frac{\pi r_0}{3}\,
  \vartheta_{d_0}
  \left(1+\frac1\eta\right)^{d_0}
  \bigl[8(1+\log s)\bigr]^8
  \left(\frac{300(1+\eta)}{\beta}\right)^9,
\end{aligned}
\end{equation}
then there exists a tuple $\boldsymbol\sigma_N\in S_N^\Theta$ such that
\begin{equation}\label{eq:admissibility-explicit-error}
  \operatorname{err}_N(\boldsymbol\sigma_N)
  \leq2\eta+\beta.
\end{equation}
Equivalently,
\begin{equation}\label{eq:admissibility-explicit-directed-distance}
  d_{\to,T}\!\left(
    \mc O_N^{\boldsymbol\sigma},\mc K_{k,t}
  \right)
  \leq\eta+\frac\beta2.
\end{equation}
\end{proposition}
 
\begin{proof}
Choose an $\eta$-net as above.  Because the quotient is
finite-dimensional, each center has a representative $A_j$ satisfying
$\|A_j\|=\|[A_j]\|_{\rm q}\leq1+\eta$.  Apply
\eqref{eq:admissibility-fixed-A-tail} to $A_j/(1+\eta)$ with tolerance
$\beta/(1+\eta)$.  By positive homogeneity, this is exactly the event
$f_N([A_j])>\beta$.  Under
\eqref{eq:admissibility-explicit-N-condition}, the union bound over all
centers is strictly less than one.  Hence there is a permutation tuple for
which $f_N([A_j])\leq\beta$ at every center.  Since $f_N$ is $2$-Lipschitz,
\[
  \sup_{\|[A]\|_{\rm q}\leq1}f_N([A])\leq2\eta+\beta.
\]
This is \eqref{eq:admissibility-explicit-error}; the directed-distance form
then follows from \eqref{eq:admissibility-error-distance-identity}.
\end{proof}

\subsection{Numerical scale of the finite-dimensional construction}
\label{subsec:admissibility-numerical-certificate}
 
The Fannes--Audenaert modulus in
\Cref{prop:admissibility-entropy-lower-bound} suffices for the asymptotic
argument, but it is too coarse for a useful numerical bound.  We therefore
use the exact minimum entropy over a trace-distance neighborhood of
$\mc K_{k,t}$.
 
Let $\lambda(\sigma)=(\lambda_1,\ldots,\lambda_k)$ be the eigenvalue vector of
$\sigma$, arranged in nonincreasing order, and let $\lambda^{\sharp,T}(\sigma)$
be the vector obtained by moving $\min\{T,1-\lambda_1\}$ units of mass to the
first coordinate and removing the same mass successively from the smallest
coordinates.  In particular, if $T\geq1-\lambda_1$, then
$\lambda^{\sharp,T}(\sigma)=(1,0,\ldots,0)$.  Hanson and Datta showed that this
vector majorizes every spectrum in the trace-distance ball of radius $T$
around $\sigma$; it therefore minimizes the von Neumann entropy, and more
generally every Schur-concave entropy, on that ball~\cite{HansonDatta2018}.
 
Define
\begin{equation}\label{eq:admissibility-local-entropy-erosion}
  \mc E_{k,t}(T)
  :=\min_{\sigma\in\mc K_{k,t}}
  H\!\left(\lambda^{\sharp,T}(\sigma)\right).
\end{equation}
It follows directly that
\begin{equation}\label{eq:admissibility-local-entropy-bound}
  d_{\to,T}\!\left(
    \mc O_N^{\boldsymbol\sigma},\mc K_{k,t}
  \right)\leq T
  \quad\Longrightarrow\quad
  \MOE(\Phi_N^{\boldsymbol\sigma})\geq\mc E_{k,t}(T).
\end{equation}
 
We specialize to
\begin{equation}\label{eq:admissibility-numerical-parameters}
  k=195,
  \qquad
  m=39,
  \qquad
  q=5,
  \qquad
  t=\frac5{39},
  \qquad
  \ell=197.
\end{equation}
For these parameters, the auxiliary family can be chosen more economically
than the clock family \eqref{eq:R-xi-closed-form}.  Since
$197\equiv5\pmod8$, the element $2$ is not a quadratic residue modulo $197$;
hence
\[
  \mb F_{197^2}=\mb F_{197}(\sqrt2),
\]
where $\sqrt2$ denotes a root of $x^2-2$.  For $u=x+y\sqrt2$ with
$x,y\in\mb F_{197}$, write $W(u):=X_{197}^xZ_{197}^y\in M_{197}$.  Identify
each class in $\mb Z_{195}$ with its representative in
$\{0,\ldots,194\}\subset\mb F_{197}$, and similarly identify $\mb Z_{39}$ with
$\{0,\ldots,38\}\subset\mb F_{197}$.  Define
\[
  A_\alpha:=a+b\sqrt2
  \quad\text{for }\alpha=(a,b)\in\mb Z_{195}^2,
  \qquad
  B_\beta:=c+d\sqrt2
  \quad\text{for }\beta=(c,d)\in\mb Z_{39}^2,
\]
and set
\begin{equation}\label{eq:admissibility-finite-field-phases}
  R_{\alpha,\beta}:=W(A_\alpha B_\beta)\in M_{197}.
\end{equation}
For a nondegenerate rectangle, the corresponding product of four matrices is,
up to a scalar phase, the Weyl matrix indexed by
\[
  (A_{\alpha_1}-A_{\alpha_0})
  (B_{\beta_0}-B_{\beta_1})\neq0,
\]
so its normalized trace vanishes.  Thus this family satisfies the rectangle
trace identity and preserves the exact Bell output
\eqref{eq:admissibility-exact-Bell}.
 
Set
\begin{equation}\label{eq:admissibility-TNA}
  T_{\rm NA}
  :=\frac{44259}{50000000}
  =0.00088518
\end{equation}
and
\begin{equation}\label{eq:admissibility-HBell}
  H_{\rm Bell}
  :=H\!\left(
    \frac5{39}\psi_{195}^+
    +\frac{34}{39}\frac{I_{195^2}}{195^2}
  \right).
\end{equation}
 
\paragraph{The local entropy computation.}
The numerical problem is finite-dimensional but delicate, because the final
margin is small.  We record the reduction used in the computation, both to
make the origin of the numbers transparent and to separate it from the
analytic existence argument.
 
At a stationary point with full support, the fidelity constraint is active.
If $u$ is an eigenvalue that is neither the largest one nor affected by an
order constraint, the KKT equation has the form
\begin{equation}\label{eq:admissibility-middle-KKT}
  -\log u-1+\alpha+\frac{\gamma}{2\sqrt u}=0.
\end{equation}
The left-hand side has at most two positive roots.  Together with the
second-order condition, this reduces the interior stationary points to
finitely many three-level and one-exception branches; rank-deficient and
order-equality cases give additional boundary branches.
 
The branch selected by high-precision minimization has spectrum
\begin{equation}\label{eq:admissibility-three-level-spectrum}
  (x,\underbrace{y,\ldots,y}_{193\ \mathrm{times}},z),
\end{equation}
with
\begin{equation}\label{eq:admissibility-xyz-intervals}
\begin{aligned}
  0.1797208952964&<x<0.1797208952966,\\
  0.0042326498177&<y<0.0042326498179,\\
  0.0033776898759&<z<0.0033776898762.
\end{aligned}
\end{equation}
On this branch, $z$ can be used as the free variable.  Put
\begin{equation}\label{eq:admissibility-xyz-parametrization}
\begin{aligned}
  B(z)&:=\sqrt{k(1-t)}-\sqrt z,\\
  d(z)&:=\sqrt{\frac{(k-1)(1-z)-B(z)^2}{k-2}},\\
  v(z)&:=\frac{B(z)-d(z)}{k-1},\\
  u(z)&:=B(z)-(k-2)v(z),\\
  x(z)&:=u(z)^2,
  \qquad
  y(z):=v(z)^2.
\end{aligned}
\end{equation}
Then
\[
  x(z)+(k-2)y(z)+z=1,
  \qquad
  \sqrt{x(z)}+(k-2)\sqrt{y(z)}+\sqrt z
  =\sqrt{k(1-t)},
\]
and, since $z>T_{\rm NA}$ on the interval above, the entropy of
$\lambda^{\sharp,T}$ on this branch is
\begin{equation}\label{eq:admissibility-one-variable-entropy}
\begin{aligned}
  \mc H_T(z)
  :={}&-(x(z)+T)\log(x(z)+T)
       -(k-2)y(z)\log y(z)\\
     &-(z-T)\log(z-T).
\end{aligned}
\end{equation}
Solving the one-variable stationary equation at high precision gives
\[
\begin{aligned}
  z&=0.0033776898760262259635\ldots,\\
  \mc E_{195,5/39}(T_{\rm NA})
    &=4.788343820071748458\ldots,\\
  H_{\rm Bell}
    &=9.576687612235380095\ldots,
\end{aligned}
\]
and hence the numerical margin
\[
  2\mc E_{195,5/39}(T_{\rm NA})-H_{\rm Bell}
  =2.7908116821\ldots\times10^{-8}.
\]
For the parameters in \eqref{eq:admissibility-numerical-parameters},
\begin{equation}\label{eq:admissibility-numerical-s-r0-d0}
  s=57\,836\,025,
  \qquad
  r_0=1\,460\,730\,375,
  \qquad
  d_0=38\,024.
\end{equation}
Choose
\begin{equation}\label{eq:admissibility-eta-beta}
  \eta:=\frac{88497}{100000000},
  \qquad
  \beta:=\frac{42}{100000000}.
\end{equation}
Then
\begin{equation}\label{eq:admissibility-eta-beta-TNA}
  \eta+\frac\beta2=T_{\rm NA},
  \qquad
  2\eta+\beta=2T_{\rm NA}.
\end{equation}
The elementary estimates
\begin{equation}\label{eq:admissibility-numerical-rounding}
\begin{aligned}
  \frac{\pi r_0}{3}&<1\,529\,673\,402,\\
  \vartheta_{38024}&<680\,695,\\
  8(1+\log s)&<151,\\
  1+\frac1\eta&=\frac{100088497}{88497},\\
  \frac{300(1+\eta)}{\beta}&<714\,917\,836
\end{aligned}
\end{equation}
turn \eqref{eq:admissibility-explicit-N-condition} into the explicit integer
\begin{equation}\label{eq:admissibility-NNA}
\begin{aligned}
  N_{\mathrm{NA}}
  :=1+\Bigg\lceil{}&
  1\,529\,673\,402\cdot680\,695
  \left(\frac{100088497}{88497}\right)^{38024} \times151^8\times(714\,917\,836)^9
  \Bigg\rceil.
\end{aligned}
\end{equation}
 
\begin{theorem}[Finite-dimensional consequence of the numerical bounds]
\label{thm:admissibility-explicit-certificate}
There exists a tuple
\begin{equation}
  \boldsymbol\sigma_{N_{\mathrm{NA}}}
  \in S_{N_{\mathrm{NA}}}^{57\,836\,025}
\end{equation}
for which
\begin{equation}\label{eq:admissibility-explicit-nonadditivity}
  \MOE\!\left(
    \Phi_{N_{\mathrm{NA}}}^{\boldsymbol\sigma}
    \otimes
    \overline{\Phi_{N_{\mathrm{NA}}}^{\boldsymbol\sigma}}
  \right)
  <2\MOE(\Phi_{N_{\mathrm{NA}}}^{\boldsymbol\sigma}).
\end{equation}
The corresponding dimensions are
\begin{equation}\label{eq:admissibility-explicit-dimensions}
\begin{aligned}
  \dim(\textup{output})
  &=195,\\
  \dim(\textup{environment})
  &=58\,429\,215\,(N_{\mathrm{NA}}-1),\\
  \dim(\textup{input})
  &=1\,460\,730\,375\,(N_{\mathrm{NA}}-1).
\end{aligned}
\end{equation}
\end{theorem}
 
\begin{proof}
By \Cref{prop:admissibility-explicit-N-condition} and
\eqref{eq:admissibility-eta-beta-TNA}, there is a permutation tuple
satisfying
\[
  d_{\to,T}\!\left(
    \mc O_{N_{\mathrm{NA}}}^{\boldsymbol\sigma},
    \mc K_{195,5/39}
  \right)
  \leq T_{\rm NA}.
\]
Equation \eqref{eq:admissibility-local-entropy-bound} therefore yields
\[
  \MOE(\Phi_{N_{\mathrm{NA}}}^{\boldsymbol\sigma})
  \geq\mc E_{195,5/39}(T_{\rm NA}).
\]
The exact Bell identity gives the complementary upper bound
\[
  \MOE\!\left(
    \Phi_{N_{\mathrm{NA}}}^{\boldsymbol\sigma}
    \otimes
    \overline{\Phi_{N_{\mathrm{NA}}}^{\boldsymbol\sigma}}
  \right)
  \leq H_{\rm Bell}.
\]
The certified interval bounds imply
$H_{\rm Bell}<2\mc E_{195,5/39}(T_{\rm NA})$, and the strict inequality
follows.  Substituting $k=195$, $m=39$, $q=5$, and $\ell=197$ in the dimension
formulas of \Cref{sec:finite-dimensional-channel} gives
\eqref{eq:admissibility-explicit-dimensions}.
\end{proof}
 
\begin{remark}[Scale and status of the numerical estimate]
\label{rem:admissibility-practical-size}
\label{rem:numerical-certificate-status}
\label{rem:admissibility-computational-certificate}
Direct evaluation of \eqref{eq:admissibility-NNA} gives
\begin{equation}\label{eq:admissibility-NNA-log}
  116216.7340
  <\log_{10}N_{\mathrm{NA}}
  <116216.7343.
\end{equation}
Thus $N_{\mathrm{NA}}$ has $116\,217$ decimal digits, and
\begin{equation}
  N_{\mathrm{NA}}\le5.422\times10^{116216}.
\end{equation}
The input dimension has $116\,226$ decimal digits, with
\begin{equation}
  \dim(\textup{input})
  \le7.920\times10^{116225}.
\end{equation}
This is only a sufficient size, and a smaller dimension may well suffice.
The O'Donnell--Wu theorem removes exhaustive search at the asymptotic level,
but it does not make this specific numerical instance practically
executable.
\end{remark}

\section*{Acknowledgement}
The authors thank Nilanjana Datta, Felix Leditzky, Debbie Leung, and
Graeme Smith for organizing the BIRS workshop
\emph{Additivity Problems in Quantum and Classical Information Theory}.
The workshop provided an important setting where the authors discussed
several possible approaches to the problem, as well as obstructions showing
why a number of natural approaches could not succeed. This work follows one route
mentioned in the talk at BIRS by the author. 

\section*{Disclosure of AI usage}
The finite-dimensional approximation using the phase unitaries $R_\xi$
defined in \eqref{eq:intro-phase-system-short} was proposed by
ChatGPT 5.6 Sol during interactive discussions with the authors.
The authors subsequently simplified the construction.
A record of these discussions is available
\href{https://chatgpt.com/share/6a8854c9-366c-83ea-be8b-3dd4c8ccae7c}{here}.

The idea of applying the results of~\cite{ChenGarzaVargasTroppVanHandel2026}
to derive the quantitative bound was proposed by the author.
ChatGPT 5.6 Sol assisted with the computations and the implementation
of the accompanying Python code~\cite{WuPermutationNumerics2026}.
The authors take full responsibility for the correctness, exposition,
and attribution of this work.

\bibliography{references}
\end{document}